\documentclass[12pt,a4paper]{article}

\usepackage[utf8]{inputenc}
\usepackage[frenchb,english]{babel}
\usepackage[T1]{fontenc}
\usepackage{lmodern}

\usepackage{dsfont,pifont,xfrac}
\usepackage{amsthm}
\usepackage{verbatim}
\usepackage{bbm}
\usepackage[left=.9in, right=.9in,top=1.2in,bottom=1.2in]{geometry}

\usepackage{graphicx}
\usepackage{subfig}
\usepackage{wrapfig,flafter,color,graphicx}
\usepackage{float}
\usepackage{enumerate,enumitem}
\usepackage[colorlinks=true]{hyperref}
\usepackage{multirow,array}   
\usepackage{tabu}
\usepackage[normalem]{ulem}

\usepackage{mathsymb-linxiao}
\usepackage{thm-linxiao}
\usepackage{twoopt}

\newcommand*{\Map}[1]{\mathfrak{#1}} 

\newcommand*{\tmap}{\Map{t}} 
\newcommand*{\emap}{\Map{e}} 
\newcommand*{\umap}{\Map{u}} 
\newcommand*{\bmap}{\Map{b}} 

\newcommand*{\bt}[1][]{(\tmap#1,\sigma#1)}
\newcommand*{\btsq}[1][]{[\tmap#1,\sigma#1]}
\newcommand*{\bts}{\mathcal{BT}}

\newcommand*{\law}{\mathcal{L}}
\renewcommand*{\cv}[2][d\1{loc}]{\xrightarrow[\raisebox{0.1em}{$\scriptstyle #2\to\infty$}]{#1}}

\newcommand*{\faces}{\mathcal{F}}
\newcommand*{\medges}{\mathcal{E}}

\newcommand*{\steps}{\mathcal{S}}
\newcommand*{\step}{\mathtt{s}}
\newcommand*{\Step}{\mathtt{S}}
\newcommand*{\CC}{\mathtt{C}}
\newcommand*{\LL}{\mathtt{L}}		\newcommand*{\RR}{\mathtt{R}}
\newcommand*{\cp}{\CC^\+}			\newcommand*{\cm}{\CC^\<}
\newcommand*{\lp}[1][k]{\LL_{#1}}	
\newcommand*{\rp}[1][k]{\RR_{#1}}	
\newcommand*{\iroot}{\mathcal I}

\newcommandtwoopt*{\weight}[4][][]{\nu^{\abs{\medges(#3#4)}#1}\, t^{\abs{\faces(#3)}#2}}
\newcommandtwoopt*{\weightc}[4][][]{\nu_c^{\abs{\medges(#3#4)}#1}\, t_c^{\abs{\faces(#3)}#2}}

\newcommand*{\zz}[2][p,q]{\,\frac{z_{#2}}{z_{#1}}}
\newcommand*{\zzz}[3][p,q]{\,\frac{z_{#2}\,z_{#3}}{z_{#1}}}
\newcommand*{\zzh}[2][p,q]{\,\frac{z_{#2}}{z_{#1}}}
\newcommand*{\zzzh}[3][p,q]{\,\frac{z_{#2}\,z_{#3}}{z_{#1}}}

\newcommand{\pqq}{{p,(q_1,q_2)}}
\newcommand{\py}[1][p]{_{#1}}
\newcommand{\yy}{_{\infty}}
\newcommand{\pqy}{_{p,q}}
\newcommand*{\+}{\ensuremath{\text{\rm\texttt{+}}}}
\newcommand*{\<}{\ensuremath{\text{\rm\texttt{-}}}}

\newcommand{\C}{\mathbb{C}}

\newcommand{\N}{\mathbb{N}}

\theoremstyle{plain}
\newcommand{\emapo}{\emap^\circ}
\newcommand{\frontier}{\partial\emap}

\newcommand*{\algo}{\mathcal{A}}

\newcommand*{\jj}{{\text{\normalsize\texttt{\textpm}}}}
\newcommand{\nseq}[2][0]{(#2_n)_{n\ge #1}}

\setlist[enumerate,1]{label=(\roman*),topsep=6pt,itemsep=0pt}
\setlist[itemize,1]{topsep=4pt,itemsep=0pt}

\newcommand{\refp}[2]{\hyperref[#2]{\ref*{#2}(#1)}}

\newcommand{\limsupp}{\limsup_{p\to\infty}}
\newcommand{\limsuppq}{\limsup_{p,q\to\infty}}
\title{Interfaces in the vertex-decorated Ising model on random triangulations of the disk}

\author{Joonas Turunen\footnote{University of Helsinki, Department of Mathematics and Statistics, joonas.am.turunen@helsinki.fi} \ \footnote{University of Iceland, Science Institute, Division of Mathematics, joonas@hi.is}}

\begin{document}

\maketitle

\begin{abstract}
We provide a framework to study the interfaces imposed by Dobrushin boundary conditions on the half-plane version of the Ising model on random triangulations with spins on vertices. Using the combinatorial solution by Albenque, Ménard and Schaeffer (\cite{AMS18}) and the generating function methods introduced by Chen and Turunen (\cite{CT20}, \cite{CT21}), we show the local weak convergence of such triangulations of the disk as the perimeter tends to infinity, and study the interface imposed by the Dobrushin boundary condition. As a consequence of this analysis, we verify the heuristics of physics literature that discrete interface of the model in the high-temperature regime resembles the critical site percolation interface, as well as provide an explicit scaling limit of the interface length at the critical temperature, which coincides with results on the continuum Liouville Quantum gravity surfaces. Overall, this model exhibits simpler structure than the model with spins on faces, as well as demonstrates the robustness of the methods developed in \cite{CT20}, \cite{CT21}. 
\end{abstract}


\newcommand*{\g}{\Map{g}}
\newcommand*{\B}{\Map{b}}


\maketitle

\section{Introduction}

Recent years have seen a great number of works where a peeling process is used to study the geometry of random planar lattices. The idea of peeling was first introduced by Watabiki in \cite{Wat95} and later made rigorous by Angel in \cite{Ang02} in the context of pure gravity. It has proven out to be a great tool to explore in particular random planar maps of the half-plane topology, to study the distances on maps, and to study the interfaces imposed by statistical physics models on them. For example, peeling techniques have been used in the context of percolation (\cite{Ang02}, \cite{Ang05}, \cite{ACpercopeel}, \cite{Rich15}), Eden model (\cite{CLGpeeling}, \cite{MS13}) and the $O(n)$ model (\cite{BudOn}).

More recently, the peeling process has also been applied to study the random triangulations of the disk coupled to the Ising model on faces by Chen and the author in the works \cite{CT20} and \cite{CT21}. There, the authors have developed a machinery based on analytic combinatorics and rational parametrizations to understand the asymptotic behavior of the partition functions, and constructed local limits using the infinite boundary limits of the perimeter processes associated with the peeling process. Meanwhile, Albenque, M\'enard and Schaeffer studied a similar model, with the exception that they considered triangulations with spins on the vertices and showed the local convergence for the full-plane topology (\cite{AMS18}). In the combinatorial part, they develop further the method of invariants introduced by Bernardi and Bousquet-Mélou in \cite{BBM11}, where one intermediate step is also the combinatorial decomposition of a triangulation with a Dobrushin boundary condition, which can also be viewed as the combinatorial definition of the one-step peeling operation. In this work, we use the combinatorial results of \cite{AMS18} to retrieve the results of \cite{CT20} and \cite{CT21} for the model with spins on the vertices. Various proofs are mutatis mutandis of the proofs found in our previous works, and thus many of the details are omitted. That said, to some extent this work can also be viewed as an expository work of the previous articles \cite{CT20}, \cite{CT21}, since it wraps up their results to consider an alternative Ising model on random triangulations. 

The advantages of the model with spins on vertices lie in the simplicity of the peeling process and the symmetry with respect to the spins, as well as in the fact that the interface is a well-defined simple curve, unlike in the case of spins on the faces. The self-duality of this model under the spin-flip also yields a critical percolation like behavior in the high-temperature regime, which differs drastically from the corresponding behaviour for the spins on faces. There, the behaviour is rather reminiscent of the subcritical face percolation. In the critical temperature, the simplicity of the interface allows us to deduce the explicit scaling limit of the interface length.

\subsection{Definition of the model}

\paragraph{Planar maps.}
A finite planar map is a proper embedding of a finite connected graph into  the sphere $\mathbb{S}^2$, viewed up to orientation-preserving homeomorphisms of $\mathbb{S}^2$.
Loops and multiple edges are allowed in the graph.
All planar maps in this work are rooted, i.e.~equipped with a distinguished oriented edge called the root edge.
In a rooted planar map $\g$, the face incident to the right of the root edge is called the \emph{external face}, and all other faces are \emph{internal faces}.
The \emph{boundary length} (or \emph{perimeter}) of $\g$ is the degree of its external face.
A (rooted) \emph{triangulation with boundary} is a rooted planar map whose internal faces are all triangles.
When the external face has no pinch-points (i.e.~the boundary is a simple path), we call it a \emph{triangulation of the $p$-gon}, where $p$ is its perimeter.
We denote by $E(\g)$ the set of edges and by $V(\g)$ the set of vertices of $\g$.
\paragraph{Ising model on the vertices of a map.} Following \cite{AMS18}, the partition function of the Ising model (without an external magnetic field) on the vertices of a map $\g$ is defined by
\begin{equation}	\label{eq:Ising}
	\mathcal{Z}(\g,\nu) = \sum_{\sigma:V(\g)\to \{\pm 1\}} \nu^{m(\g,\sigma)}
\end{equation}
where $\nu>0$ is the coupling constant, $\sigma$ represent the spin configuration and $m(\g,\sigma)$ is the number of monochromatic edges (edges which share the same spin in the endpoints).
Moreover, we consider the Dobrushin boundary conditions, that is, the spins on the boundary vertices are fixed by a sequence of the form $+^p-^q$ ($p$ $+$'s followed by $q$ $-$'s) in the counter-clockwise order from the root. Throughout this work, we call the root edge $\rho$ and the other bichromatic boundary edge opposite to the root $\rho^\dagger$.

\begin{figure}
\begin{center}
\includegraphics[scale=0.7]{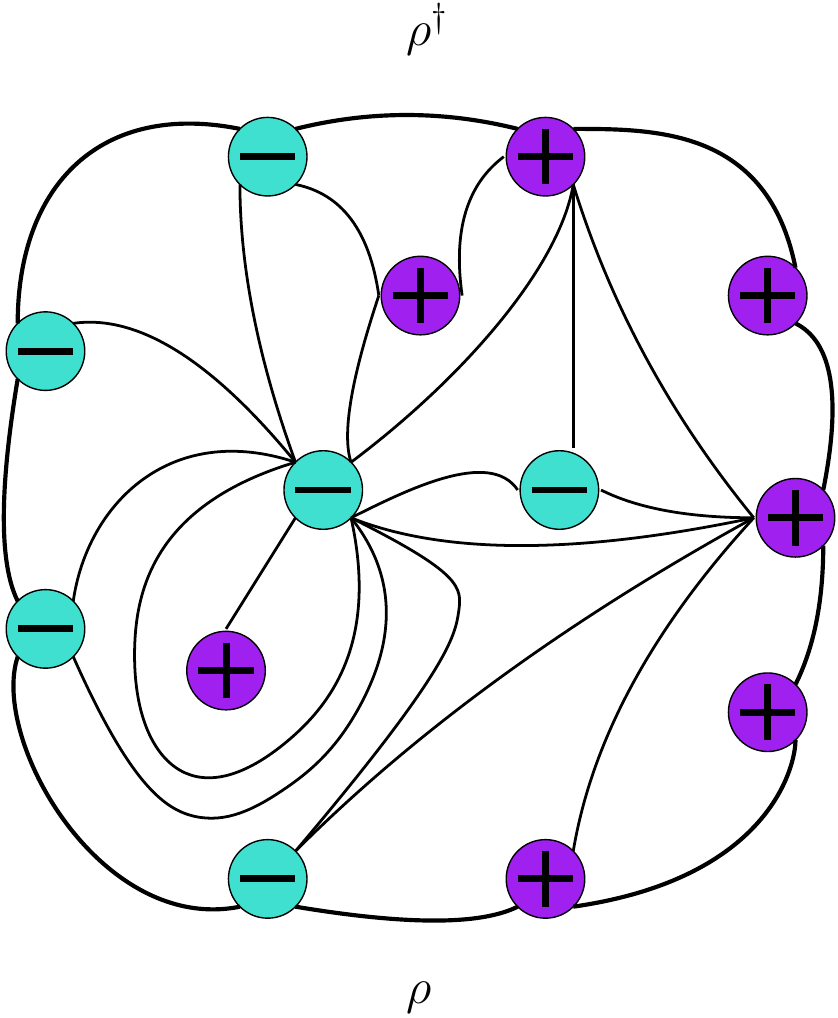}
\end{center}\vspace{-1em}
\caption{An example of a vertex-decorated Ising triangulation of the $p+q$-gon with Dobrushin boundary condition $(p,q)=(5,4)$, where $\rho$ denotes the root edge and $\rho^\dagger$ the other extremity. The triangulation has $27$ edges and $17$ monochromatic edges.}
\end{figure}

\paragraph{Ising triangulations.}

Let $\mathcal{G}_{p,q}$ be the set of rooted planar triangulations of the $(p\!+\!q)$-gon, endowed with the Dobrushin boundary condition of type $+^p-^q$, which we will call \emph{triangulations of the $(p,q)$-gon}.
From a combinatorial point of view, if $\sigma$ is a spin configuration on $V(\g)$, then the pair $(\g,\sigma)$ is just a vertex-bicolored map.
Let us denote by $\mathcal{G}^\sigma_{p,q}$ the set of vertex-bicolored triangulations of the $(p,q)$-gon.

For $\nu>0$, let
\begin{equation}\label{eq:partitionfunction}
	z_{p,q}(t,\nu)  = \sum_{\g\in \mathcal{G}_{p,q}} \mathcal{Z}(\g,\nu)\, t^{\abs{E(\g)}}
					= \sum_{(\g,\sigma)\in \mathcal{G}^\sigma_{p,q}} \nu^{m(\g,\sigma)} t^{\abs{E(\g)}}.
\end{equation}This is called the \emph{partition function} of the Ising-triangulation of the $(p,q)$-gon. In the end, we are interested in the asymptotics of this quantity as $p,q\to\infty$, giving information about large Ising-triangulations with a boundary.
In order to encode the partition functions, define first the generating series of Ising-triangulations with a monochromatic boundary (also known as \emph{disk amplitude}), by
\begin{equation}\label{eq: Z0}
Z_0(u;t,\nu):=\sum_{p\ge 1} z_{p,0}(t,\nu)\, u^p.
\end{equation}
Then, define the generating series of $(z_{p,q})_{p,q\ge 1}$:
\begin{equation}\label{eq: Z}
	Z(u,v;t,\nu) = \sum_{p,q\ge 1} z_{p,q}(t,\nu)\, u^p v^q.
\end{equation} Observe that $Z_0(u)$, unlike in the case of spins on the faces, cannot be recovered from $Z(u,v)$ via coefficient extraction. Since this might appear to look like a problem when applying the singularity analysis methods of \cite{CT20} and \cite{CT21}, we will also consider the generated function
\begin{equation}\label{eq:completedGF}
\mathring{Z}(u,v;t,\nu):=\sum_{p+q\ge 1} z_{p,q}(t,\nu)\, u^p v^q=Z(u,v;t,\nu)+Z_0(u;t,\nu)+Z_0(v;t,\nu).
\end{equation} Let also $Z_q(u;t,\nu) = [v^q]Z(u,v;t,\nu)$ for $q\geq 1$. 

For each $\nu>0$, let $t_c(\nu)$ be the radius of convergence of the series $t\mapsto z_{1,0}(t,\nu)$.
In \cite[Theorem 6]{AMS18}, it was shown that $t_c(\nu)$ is the unique dominant singularity of $t\mapsto z_{p,q}(t,\nu)$ for all $p\ge 0$, $q\ge 0$, such that $p+q\geq 1$. Moreover, it was shown that $z_{p,q}(t,\nu)<\infty$ if $|t|\le t_c(\nu)$.
In the sequel, we will only consider parameters $(t,\nu)$ on the \emph{critical line} $t=t_c(\nu)$. Doing so, we may omit $t$ from the list of variables, and write $z_{p,q}(\nu) \equiv z_{p,q}(t_c(\nu),\nu)$, $Z(u,v;\nu) \equiv Z(u,v;t_c(\nu),\nu)$, and so on. The finiteness of $z_{p,q}(\nu)$ allows us to consider the Boltzmann distribution, which is a special case of the one in \cite[Definition 22]{AMS18}, defined for more general boundary conditions.

\begin{defi}\label{def:BIT}
The Boltzmann Ising-triangulation of the $(p,q)$-gon is the law $\prob_{p,q}^\nu$ defined by
\begin{equation}
	\prob_{p,q}^\nu(\tmap,\sigma) = \frac{ \nu^{m(\tmap,\sigma)} t_c(\nu)^{\abs{E(\tmap)}} }{z_{p,q}(\nu)}
\end{equation}
for all $(\tmap,\sigma)\in \mathcal{G}^\sigma_{p,q}$. 
\end{defi} 

The local distance between Ising-decorated triangulations (or maps in general) is defined by
\begin{equation*}
d\1{loc}(\bt,\bt[']) = 2^{-R}\qtq{where}
	R = \sup\Set{r\geq 0}{ \btsq_r=\btsq[']_r }
\end{equation*}
and $\btsq_r$ denotes the ball of radius $r$ around the origin in $\bt$ which takes into account the spins of the vertices. The set $\bts$ of (finite) vertex-bicolored triangulations of polygon is a metric space under $d\1{loc}$. Let $\overline{\bts}$ be its Cauchy completion. Recall that an (infinite) graph is \emph{one-ended} if the complement of any finite subgraph has exactly one infinite connected component. It is well known that a one-ended map has either zero or one face of infinite degree \cite{CurPeccot}. We call an element of $\overline{\bts}\setminus \bts$ a \emph{vertex-bicolored triangulation of the half plane} if it is one-ended and its external face has infinite degree. Namely, such a triangulation has a proper embedding in the upper half plane without accumulation points and such that the boundary coincides with the real axis. We denote by $\bts_\infty$ the set of all vertex-bicolored triangulations of the half plane.

\subsection{Main results}

We obtain the local convergence, i.e. the convergence in distribution w.r.t. the local distance, of vertex-decorated Ising-triangulations as the perimeter of the disk tends to infinity, for high temperatures and at the critical point:

\begin{theorem}[Local limits of Boltzmann Ising-triangulations]~\label{thm:cv}\\
For every $1<\nu\leq\nu_c=1+1/\sqrt{7},$ there exists a probability distribution $\prob_\infty^\nu$, such that for all $0<\lambda_{\min}\leq 1\leq \lambda_{\max}<\infty$,
\begin{equation*}
\prob_{p,q}^\nu\ \cv{p,q}\ \prob_\infty^\nu\qquad\text{while}\quad\frac{q}{p}\in[\lambda_{\min},\lambda_{\max}]
\end{equation*}
locally in distribution. We also have 
\begin{equation*}
\prob_{p,q}^{\nu_c}\ \cv{q}\ \prob\py^{\nu_c}\ \cv{p}\ \prob_\infty^{\nu_c}.
\end{equation*}
 The laws $\prob_\infty^\nu$ and $\prob\py^{\nu_c}$ are supported on $\bts_\infty$.
\end{theorem}

\begin{remark}
The above local convergence could be proven for $\nu>\nu_c$ too, with apparently more complicated rational parametrizations which require cumbersome technical details for being simplified. We do not do it here, since we expect a trivial interface structure in the form of a bottleneck. See the analogous case for the spins in the faces in \cite{CT21}. Moreover, the two-step local limit should also hold for $1<\nu<\nu_c$ (in fact for all $\nu>1$), but it is not proven here, since it requires slightly more technical lemmas and thus provides little value for this work. We also leave the treatment of the antiferromagnetic regime $0<\nu<1$ for future work.
\end{remark}

The proof of Theorem~\ref{thm:cv} spans throughout Sections~\ref{sec:combi}-\ref{sec:locallimit}, and roughly consists of the following three parts. First, Section~\ref{sec:combi} contains the asymptotic analysis of the generating functions \eqref{eq: Z0} and \eqref{eq: Z}, culminating in asymptotic formulas of the partition functions \eqref{eq:partitionfunction} as $p,q\to\infty$ similar to those of \cite[Theorem 2]{CT21}. There, the combinatorial methods of \cite{CT20} and \cite{CT21} are applied starting from the algebraic equations derived in \cite{AMS18}. Then, these formulas allow us to construct the limits of the \emph{peeling process} following the Ising interface $\iroot$. The peeling process reveals one triangle adjacent to $\iroot$ at each step, and swallows a finite number of other triangles if the revealed triangle separates the unexplored part in two pieces. Formally, it is defined as an increasing sequence of submaps $\nseq \emap$, whose precise definition is left to Section~\ref{sec:peeling}. Alternatively, it can also be defined via a sequence of \emph{peeling events} $\nseq[1] \Step$ taking values in a countable set of symbols, where $\Step_n$ indicates the position of the triangle revealed at time $n$ relative to the explored map $\emap_{n-1}$. See Section~\ref{sec:peeling} for a detailed account. The central feature is that the law of the sequence $\nseq[1] \Step$ can be written down easily and one can perform explicit computations with it. More notably, it possesses limits in distribution as $p,q\to\infty$, which in turn can be used in the construction of the peeling process in the limit. The limits of the peeling process are then used to construct the local limits and to show the local convergences in Section~\ref{sec:locallimit}.

Assume now $\nu=\nu_c$. Let $\eta_{p,q}$ be the length of the interface between the edges $\rho$ and $\rho^\dagger$ in an Ising-triangulation $\bt$ sampled from $\prob_{p,q}\equiv\prob_{p,q}^{\nu_c}$. Similarly, let $\eta_p$ be the length of the interface in $\bt$ sampled from $\prob_p\equiv\prob_p^{\nu_c}$. The main theorem of this work comprises the following scaling limits of the interface length:

\begin{theorem}\label{thm:scaling}
Let $\nu=\nu_c$, and $\mu:=\frac{11-5\sqrt{7}}{12\sqrt{7}-48}>0$. Then
\begin{equation*}
\forall t>0\,,\qquad	\lim_{p,q\to\infty} \prob_{p,q}\m({\mu \eta_{p,q}>tp} = \int_t^\infty(1+s)^{-7/3}(\lambda+s)^{-7/3}ds
\end{equation*}where the limit is taken such that $q/p\to\lambda\in (0,\infty)$. In particular, for $\lambda=1$, \begin{equation*}
\lim_{p,q\to\infty} \prob_{p,q}\m({\eta_{p,q}>tp} =(1+\mu t)^{-11/3}.
\end{equation*} Moreover, we have 
\begin{equation*}
\forall t>0\,,\qquad	\lim_{p\to\infty} \prob\py\m({\eta_p>tp} = (1+\mu t)^{-4/3}.
\end{equation*}
\end{theorem}

\begin{remark}
The limit law $\prob(L>t):=\int_t^\infty(1+x)^{-7/3}(\lambda+x)^{-7/3}dx$ has an interpretation in the continuum Liouville Quantum Gravity as the length of the gluing interface of two independent quantum disks of the LQG of parameter $\gamma=\sqrt{3}$ along suitable scaled boundary segments (see \cite{matingoftrees}, \cite{DKRV16}). This is explained in detail in \cite{CT21} and the references therein. Similarly, the limit law $\prob(M>t):=(1+\mu t)^{-4/3}$ is the length of the gluing interface of a quantum disk together with a thick quantum wedge along a suitable boundary segment. This is demonstrated in \cite{CT20}. To summarize, in both cases the limit laws match the predictions from the mating of trees theory of LQG, suggesting that the scaling limit of the interface coupled to a suitable conformal structure should be an SLE$(3)$ curve on a LQG surface.
\end{remark}

The proof of Theorem~\ref{thm:scaling} is based on various asymptotic properties of the \emph{perimeter processes} associated with the peeling process. In particular, the interface length $\eta_{p,q}$ or $\eta_p$ has the same scaling limit as a hitting time of the perimeter processes in a neighborhood of the origin. Since the asymptotic properties of the perimeter processes are the same as in the setting where the spins are on the faces (\cite{CT20}, \cite{CT21}) up to constants, the techniques introduced in that context will also work here. In addition, explicit computations show that the scaling exponents in Theorem~\ref{thm:scaling} indeed coincide with those found in \cite{CT20}, \cite{CT21}. The combinatorial interpretations of this phenomenon remain by now unexplained. 

\paragraph{Outline.} Section~\ref{sec:combi} analyses and collects the combinatorial results which are used in order to find the asymptotics of the partition functions. More specifically, it explains briefly the derivation of the functional equations of \cite{AMS18} used in this work, presents the rational parametrizations of the associated generating functions, contains the singularity analysis of these generating functions and states the asymptotics which can be derived from parametrizations sharing a similar singularity structure as in \cite{CT21}.

The rest of the sections focus on the probabilistic analysis of the model. More precisely: Section~\ref{sec:peeling} contains the definition and the basic properties of the peeling process as well as defines the associated perimeter processes, whose asymptotic properties are gathered in Section~\ref{sec:asympt}. The local limits are constructed in Section~\ref{sec:locallimit}, which also contains the proof of Theorem~\ref{thm:cv}. Lastly, Section~\ref{sec:scalingl} is devoted to the proof of Theorem~\ref{thm:scaling}.

\paragraph{Acknowledgements.}

This work has been benefited from the author's joint research with L. Chen, whom the author acknowledges for great influence and communication especially for the combinatorics. The author also thanks M. Albenque, K. Izyurov, A. Kupiainen and L. Ménard for insightful discussions. The work has been supported by the ERC Advanced Grant 741487 (QFPROBA) as well as by the Icelandic Research Fund, Grant Number: 185233-051.

\section{Combinatorics of the model}\label{sec:combi}

\subsection{Recursive decomposition via peeling}

\begin{figure}
\begin{center}
\includegraphics[scale=0.8]{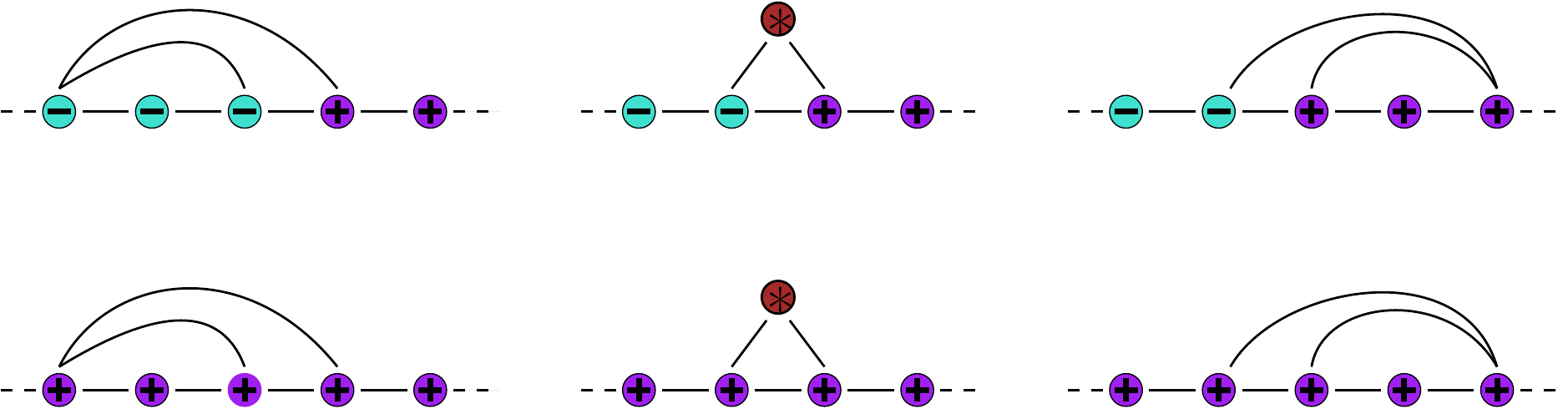}
\end{center}\vspace{-1em}
\caption{A schematic figure explaining the peeling decomposition. It lists the three different locations of the third vertex of the revealed triangle from the left to the right. In the center, the third vertex lies in the interior and has a spin $*\in\{\<,\+\}$. The first row depicts the peeling with a Dobrushin boundary and the second with a monochromatic boundary.}
\end{figure}

The peeling decomposition is considerably simpler than in the case of spins on the faces. We recall its construction from \cite{AMS18}. The decomposition is purely deterministic given an Ising-triangulation, and together with the Boltzmann distribution it also formally defines the first step in the peeling process itself. 

We consider a vertex-bicolored triangulation $\bt$ with a boundary condition $p,q\ge 1$, and delete the root edge separating two opposite boundary spins. If $\bt$ has an interior face, then the third vertex of the revealed triangle is either an interior vertex, in which chase both of the spins have equal probability, or a boundary vertex, whose spin is completely defined by its location. This gives us the recursive equation \begin{equation}\label{receq1}
z_{p,q}=t \left(z_{p+1,q}+z_{p,q+1}+\sum_{k=0}^{p-1}z_{p-k,q}z_{k+1,0}+\sum_{k=0}^{q-1}z_{p,q-k}z_{0,k+1}+\delta_{p,1}\delta_{q,1}\right)
\end{equation}where the coefficient $t$ results from deleting an edge and the last term corresponds to a triangulation consisting of a single edge with two vertices. Moreover, if $q=0$, we have the recursion \begin{equation}\label{receq2}
z_{p,0}=t\nu \left(z_{p+1,0}+z_{p,1}+\sum_{k=0}^{p-1}z_{p-k,0}z_{k+1,0}+\delta_{p,2}\right)
\end{equation}where this time we have the coefficient $\nu$ from the deletion of a monochromatic boundary edge. Now summing the recursion relations \ref{receq1} and \ref{receq2} over $p,q\ge 0$ yields a closed system of functional equations \begin{equation}\label{eq:masterFE}
\small\begin{cases}Z(u,v)=t\left(\frac{1}{u}\left(Z(u,v)-uZ_1(v)+Z(u,v)Z_0(u)\right)+\frac{1}{v}\left(Z(u,v)-vZ_1(u)+Z(u,v)Z_0(v)\right)+uv\right) \\
Z_0(u)=t\nu\left(\frac{1}{u}\left(Z_0(u)-uz_{1,0}\right)+\frac{1}{u}Z_0(u)^2+Z_1(u)+u^2\right)
\end{cases}.
\end{equation}This system is solved in \cite{AMS18} using a generalization of the kernel method. As an outcome \cite[Theorem 13]{AMS18}, an equation of one catalytic variable for the disk amplitude is obtained:
\begin{equation}\label{peq1cat}
2t^2\nu(1-\nu)\left(\frac{Z_0(u)}{u}-z_{1,0}\right)=u\cdot\text{Pol}\left(\nu,\frac{Z_0(u)}{u},z_{1,0},z_{2,0},t,u\right)
\end{equation}where $\text{Pol}(\nu,a,a_1,a_2,t,u)$ is an explicit polynomial. For us, the starting point is its equivalent form \cite[(24)]{AMS18},
\begin{align}\label{peq1cat2}
\small 0&=2\nu(\nu-1)t^2Z_0(u)^3+(\nu^2t^3u^2-(\nu(\nu+1)-2)tu+4t^2\nu(\nu-1))Z_0(u)^2+(-2t^2\nu(\nu-1)uz_{1,0}\notag\\
&+\nu(2\nu-3)t^2u^3+(2\nu^2t^3+\nu-1)u^2-(\nu(\nu+1)-2)tu+2t^2\nu(\nu-1))Z_0(u)\notag\\
&-2t^2\nu(\nu-1)u^2z_{1,0}^2+(-2u^3t^3\nu^2+(\nu(\nu+1)-2)tu^2-2\nu(\nu-1)t^2u)z_{1,0}\notag\\
&-2\nu(\nu-1)t^2u^2z_{2,0}+u^5t^3\nu^2-\nu(\nu-1)tu^4+\nu(\nu-1)t^2u^3.
\end{align}

\subsection{Solutions of the generating functions as rational parametrizations}

We make the changes of variables $t^3\to T$, $tZ_0(u)/u\to F$ and $t^2u \to U$ as well as $tz_{1,0}\to\mathcal{Z}_1$ and $t^2z_{2,0}\to\mathcal{Z}_2$, which give the following expression equivalent with \eqref{peq1cat2}: 
\begin{align}\label{peq1cat3}
0&=2\nu(\nu-1)U^2F^3+(\nu^2U^3-(\nu(\nu+1)-2)U^2+4\nu(\nu-1)TU)F^2+(-2\nu(\nu-1)\mathcal{Z}_1TU\notag\\ &+\nu(2\nu-3)U^3+2\nu^2TU^2+(\nu-1)U^2-(\nu(\nu+1)-2)TU+2\nu(\nu-1)T^2)F\notag\\
&-2\nu(\nu-1)\mathcal{Z}_1^2TU+(-2\nu^2TU^2+(\nu(\nu+1)-2)TU-2\nu(\nu-1)T^2)\mathcal{Z}_1\notag\\
&-2\nu(\nu-1)\mathcal{Z}_2TU+\nu^2U^4-\nu(\nu-1)U^3+\nu(\nu-1)TU^2.
\end{align}The works \cite{BBM11} and \cite{AMS18} provide rational parametrizations for $T$, $\mathcal{Z}_1$ and $\mathcal{Z}_2$. In particular, as in \cite[Theorem 23]{BBM11} and \cite[Theorem 7]{AMS18}, 
\begin{equation}\label{eq:paramT}
\small T=\hat{T}(\nu,S):=S\frac{\left((1+\nu)S-2\right)\left(8(\nu+1)^2S^3-(11\nu+13)(\nu+1)S^2+2(\nu+3)(2\nu+1)S-4\nu\right)}{32\nu^3(1-2S)^2}
\end{equation}where $S=S(\nu,t)$ is the unique power series in $t^3$ with constant term zero and satisfying the above equation. The lengthy parametrizations $\hat{\mathcal{Z}_1}(\nu,S)$ and $\hat{\mathcal{Z}_2}(\nu,S)$ of $\mathcal{Z}_1$ and $\mathcal{Z}_2$, respectively, are given in the Maple worksheet \cite{CAS3}.

Plugging the aforementioned parametrizations in equation \eqref{peq1cat3} yields an algebraic equation of genus zero in the variables $U$ and $F=F(U)$. This fact is easily verified with computer algebra, and we only need the direct consequence that there exists a rational parametrization $(U,F)=(\hat{U}(H;S,\nu), \hat{F}(H;S,\nu))$ in a complex variable $H$, as well as to find one which is simple enough for concrete algebraic manipulations. After this, we obtain a rational parametrization for $(u,v,Z(u,v))$, since the master functional equation system \eqref{eq:masterFE} has the solution \begin{equation}\label{eq:masterFEsol}
\begin{cases}Z(u,v)=\frac{U(u)^2U(v)^2-TU(u)U(v)(tZ_1(u)+tZ_1(v))}{TU(u)U(v)(1-F(u)-F(v))-T^2(U(u)+U(v))} \\
tZ_1(u)=\frac{1}{\nu T}\left(U(u)F(u)(1-\nu F(u))-\nu U(u)^2-\nu T(F(u)-\mathcal{Z}_1)\right)
\end{cases}.
\end{equation}Above, we denote $U(x):=t^2x$ and $F(x):=tZ_0(x)/x$, and omit the parameters $S$ and $\nu$ for simplicity. However, these parametrizations are in general complicated, and thus we want to eliminate either $S$ or $\nu$. A natural and physically justified way to do this is to restrict to the \emph{critical line} of the model (see \cite{CT21}).

\subsection{Critical line}

Recall that the critical temperature of the model is $\nu_c=1+\frac{\sqrt{7}}{7}$. Let $\nu\in (0,\nu_c]$ be fixed and $S$ be a free complex parameter. It is easy to see that the rational parametrization $(\hat{T},\hat{\mathcal{Z}}_1)$ is real and proper in $S$, thus amenable for the singularity analysis techniques of Appendix B in \cite{CT20}. 

We see that if $S=0$ then $(\hat{T},\hat{\mathcal{Z}}_1)=(0,0)$. Thus, $(\hat{T},\hat{\mathcal{Z}}_1,\nu)$ parametrizes $\mathcal{Z}_1$ locally at $T=0$. Next, we want to find the critical points of this parametrization. They are exactly the values of $S$ such that $|\hat{\mathcal{Z}}_1(\nu,S)|=\infty$ or $\frac{\partial}{\partial S}\hat{T}(\nu,S)=0$. The former cannot hold, since we notice that the condition $|\hat{\mathcal{Z}}_1(\nu,S)|=\infty$ implies $\hat{T}(\nu,S)=0$, which is not true since the radius of convergence of $\mathcal{Z}_1$ is positive (see \cite{AMS18}). The solutions of $\frac{\partial}{\partial S}\hat{T}(\nu,S)=0$ are precisely the solutions to the equation $(3S^2-3S+1)\nu+3S^2-3S=0$, giving the unique solution $\nu=\frac{3S(1-S)}{3S^2-3S+1}$. This identity defines a continuous bijection $S\mapsto\nu(S)$ from $(0,S_c]$ to $(0,\nu_c]$, where $S_c:=\inf\{S>0 | \frac{\partial}{\partial S}\hat{T}(\nu,S)=0\}=\frac{5-\sqrt{7}}{9}$ is the parameter of a dominant singularity of $\mathcal{Z}_1$. Indeed, since $\mathcal{Z}_1$ has positive coefficients as a combinatorial generating function, and furthermore $\hat{T}(\nu,\infty)=\infty$, this result follows from Proposition 21 in \cite{CT20}. 

Now for any $\nu>0$, the critical value of $T$ is $T_c(S):=\hat{T}(S, \nu(S))$, where $\nu(S)=\frac{3S(1-S)}{3S^2-3S+1}$ for $0<S\leq S_c$. We call the curve $(\nu, T)$ parametrized by the aforementioned parameter values $(S,T_c(S)$ the \emph{critical line}, and continue to work on it in the sequel. We stress that $S$ is \emph{not} a physical temperature parameter. In \cite{AMS18}, the critical line in the physical temperature parameter $\nu$ has been given the equation \begin{displaymath}27648t^6\nu^4+864\nu(\nu-1)(\nu^2-2\nu-1)t^3+(7\nu^2-14\nu-9)(\nu-2)^2=0,\end{displaymath} which has the first branch solution (see \cite{AMS18} for its derivation) \begin{displaymath}
t^3(\nu)=T(\nu)=\frac{1}{576}\frac{-9\nu^3+27\nu^2+\sqrt{3}\sqrt{-(\nu^2-2\nu-3)^3}-9\nu-9}{\nu^3}.
\end{displaymath}The corresponding expression of the critical line in $S$ is 
\begin{equation}\label{eq:criticalline}
T_c(S)=-\frac{1}{864}\frac{(6S^2-10S+3)(3S-2)^2}{S(S-1)^3}.
\end{equation}


\subsection{Singularity analysis at $\nu=\nu_c$}

\paragraph{The singular behavior of the generating functions.} At $\nu=\nu_c$, we have $S_c=\frac{5-\sqrt{7}}{9}$, giving $t_c:=(T_c(S_c))^{1/3}=\frac{25\sqrt{7}-55}{864}$. Applying the \emph{parametrization} method of the \emph{algcurves}-package of Maple to equation \eqref{peq1cat3}, we find some rational parametrization $(\hat{U}(R),\hat{F}(R))$, for which we can apply Möbius transformations in order to move singularities and simplify its expression (see \cite{CAS3}). More precisely, we find it convenient to make the following change of variable: we move the value $R_0$ parametrizing $F$ at the origin to $0$ and the parameter of the dominant singularity $R_c$ to $1$ using the transformation $R=R_0-(R_0-R_c)H$ (actually, we might just guess the values of $R_0$ and $R_c$, and check that the corresponding values in the new parametrization are indeed as required). Then, defining $\hat{u}(H):=t_c^{-2}\hat{U}(H)$ and $\hat{Z}_0(H):=t_c^{-1}\hat{u}(H)\hat{F}(H)$ gives the following parametrization:
\begin{align*}
&\hat{u}(H)=\frac{4(\sqrt{7}-4)}{(50\sqrt{7}-110)^{2/3}}H\frac{2H^2-6H+5}{H-2}
\end{align*}
\begin{align}
\hat{Z}_0(H)&=\frac{(\sqrt{7}-1)(\sqrt{7}-4)}{5(5\sqrt{7}-11)}\notag\\ &\cdot H\frac{8H^4+(4\sqrt{7}-44)H^3+(96-20\sqrt{7})H^2+(34\sqrt{7}-101)H+44-20\sqrt{7}}{(H-2)^2}\notag\end{align}
\begin{align}
\hat{Z}(H,K)&=\notag\\
&-\frac{8}{5}HK\Bigg(\frac{(K^3-5K^2+\frac{17}{2}K-5)H^3+(-5K^3+24K^2-\frac{313}{8}K+22)H^2}{(H-2)^2(K-2)^2(H+K-2)}\notag\\
&+\frac{(\frac{17}{2}K^3-\frac{313}{8}K^2+\frac{245}{4}K-33)H-5K^3+22K^2-33K+17}{(H-2)^2(K-2)^2(H+K-2)}\Bigg).
\end{align} We also define $\check{Z}(H,K):=\hat{Z}(H,K)+\hat{Z}_0(H)+\hat{Z}_0(K)$. We notice that $H=0$ parametrizes the generating functions $u\mapsto Z_0(u)$ and $u\mapsto Z(u,v)$ at the origin. Moreover, an explicit computation (\cite{CAS3}) of the derivative of $\hat{u}$ shows that the smallest positive critical point of $\hat{u}$ is $H=1$ (the other one being $5/2$), and this critical point is a double zero of $\hat{u}'(H)$. We define $u_c:=\hat{u}(1)=\frac{4(4-\sqrt{7})}{(50\sqrt{7}-110)^{2/3}}$. In the following, we drop the notion of $\nu=\nu_c$ from the arguments of the generating functions for simplicity.

\begin{lemma}\label{lem:absconvcrit}
The series $(u,v)\mapsto \mathring{Z}(u,v)=Z_0(u)+Z_0(v)+Z(u,v)$ is absolutely convergent if and only if $|u|\le u_c$ and $|v|\le u_c$.
\end{lemma}

\begin{proof}
First, we note that $\mathring{Z}(u,0)=Z_0(u)$, having the parametrization $(\hat{u}(H),\hat{Z}_0(H))$. Then, since $\hat{Z}_0(H)$ only has pole at $H=2>1$, it follows that $u_c$ is the radius of convergence of $Z_0(u)$; we use the argument of the proof of \cite[Lemma 5]{CT20}. Moreover, $\hat{Z}_0(H)$ is finite at $H=1$, which yields the convergence of $Z_0(u)$ at $u=u_c$. We also notice that $\check{Z}(H,H)$ does not have a pole for which $|H|\geq 1$, since an explicit computation shows that the only candidate $H=1$ appears to be a removable singularity. Thus, $\mathring{Z}(u_c,u_c)<\infty$, and the rest follows from the proof of \cite[Lemma 5]{CT20}.
\end{proof}
 
\begin{remark}
The reason why we consider the series $\mathring{Z}(u,v)$ instead of $Z(u,v)$ is the fact that $Z(u,v)$ also converges for $(u,0)$ with $|u|>u_c$, realizing the value zero. This is only a technicality, encapsulating all the essential generating series for the coefficients of the partition function in a single bivariate generating function. This trick also makes the singularity analysis to fall within the framework of \cite{CT20}, \cite{CT21}.
\end{remark}

For convenience and notational consistency with \cite{CT21}, we make the change of variables $(x,y)=\left(\frac{u}{u_c},\frac{v}{u_c}\right)$ and define $\tilde{Z}(x,y):=\mathring{Z}(u_c x,u_c y)$. The corresponding rational parametrization is $x=\hat{x}(H):=u_c^{-1}\hat{u}(H)$ and $\tilde{Z}(x,y)=\check{Z}(H,K)$. Now $\hat{x}$ induces a conformal bijection from a neighborhood of $H=0$ to the unit disk $\mathbb{D}$, which extends continuously to the boundary of $\mathbb{D}$. Let $\mathcal{H}_0$ be the component of the preimage of $\hat{x}^{-1}(\mathbb{D})$ containing the origin, and let $\overline{\mathcal{H}}_0$ be its closure. We readily obtain the following two lemmas highlighting the singularity structure of the rational parametrizations in $\overline{\mathcal{H}}_0$.

\begin{lemma}\label{lem:ucritc}
The value $H=1$ is the unique critical point of $\hat{x}$ in $\overline{\mathcal{H}}_0$, being of multiplicity $2$. 
\end{lemma}

\begin{proof}
The zeros of $\hat{x}'(H)$ coincide with the zeros of $\hat{u}'(H)$, being exactly $H=1$ and $H=5/2$. In particular, they are both on the positive real line, and the former is a double zero. The latter cannot belong to $\overline{\mathcal{H}}_0$, since the domain $\mathcal{H}_0$ has the topology of the disk and is symmetric with respect to the $H$-axis. (See the proof of \cite[Lemma 13]{CT21} for a similar statement).
\end{proof}

\begin{lemma}\label{lem:upolec}
The value $(H,K)=(1,1)$ is the unique pole of $\check{Z}$ in $\overline{\mathcal{H}}_0^2$.
\end{lemma}

\begin{proof}
From the expressions of $\hat{Z}(H,K)$ and $\hat{Z}_0(H)$, we see that the possible poles are located at $H=2$, $K=2$ or $H+K=2$. Now $H=2$ cannot belong to $\overline{\mathcal{H}}_0$, since the domain $\mathcal{H}_0$ has the topology of the disk and is symmetric with respect to the $H$-axis. By symmetry, the same holds for $K=2$. Finally, let us assume that $(H,K)\in\overline{\mathcal{H}}_0^2$ such that $H+K=2$. Since $\tilde{Z}(x,y)$ is absolutely convergent in $\overline{\mathbb{D}}$, then necessarily the numerator of $\check{Z}$ vanishes for such $(H,K)$. Substituting $K=2-H$ gives the numerator an expression proportional to $(H-1)^6$, which only vanishes at $H=1$. Thus, $(H,K)=(1,1)$ is the unique pole.
\end{proof}

\begin{figure}%
    \centering
    \subfloat[The $\Delta$-domain $\Delta_{\epsilon,\theta}$.]{{\includegraphics[scale=0.8]{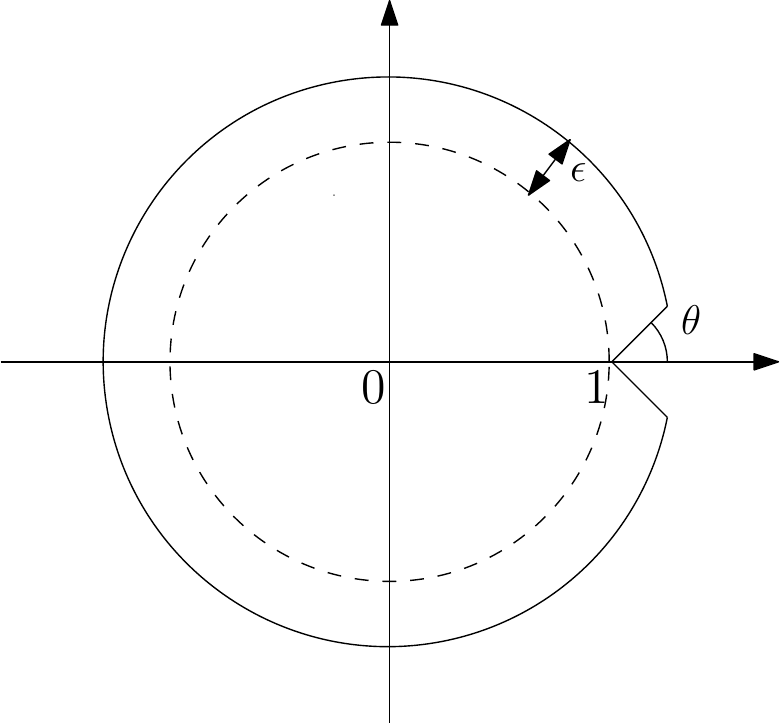} }}%
    \qquad
    \subfloat[The slit disk $\Delta_\epsilon$.]{{\includegraphics[scale=0.8]{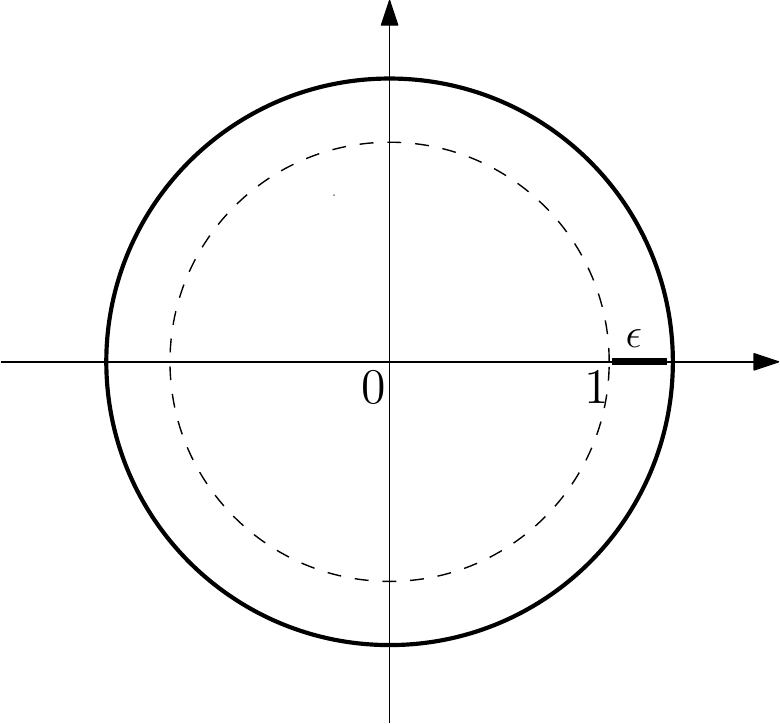}}}%
   \caption{The construction of the $\Delta$-domain and the slit disk.}%
    \label{fig:SLE}%
\end{figure}

\begin{definition}
Let $\epsilon>0$ and $\theta\in\left(0,\frac{\pi}{2}\right)$. The domain \begin{displaymath}\Delta_{\epsilon,\theta}:=\{z\in (1+\epsilon)\mathbb{D}\ |\ z\ne 1\text{ and }|\arg(z-1)|>\theta\}\end{displaymath} is called the \emph{$\Delta$-domain} of opening angle $\theta$ and margin $\epsilon$ at $1$. 
For $\theta=0$, the $\Delta$-domain reduces to \begin{displaymath}\Delta_\epsilon:=(1+\epsilon)\mathbb{D}\setminus [1,\epsilon).\end{displaymath} It is called the \emph{slit disk} of margin $\epsilon$ at $1$.
\end{definition}

Holomorphicity in a $\Delta$-domain is a basic condition for the transfer theorems for coefficient asymptotics in analytic combinatorics; see \cite{FS09} for theory and examples. In this special case $\nu=\nu_c$, we can actually show holomorphicity in a product of slit disks, which ultimately follows from the fact that the parameter of the dominant singularity is a second order critical point of $\hat{x}$. See \cite[Section 3]{CT21} for the complete explanation.

We make the following topological convention regarding the closure of the $\Delta$-domains: if $\theta>0$,  the closure $\overline{\Delta}_{\epsilon,\theta}$ is defined with respect to the usual Euclidean topology of $\C$, whereas if $\theta=0$, the closure of the slit disk $\overline{\Delta}_{\epsilon}$ is defined with respect to the topology of the universal covering space of $\C\setminus\{1\}$. The latter ensures that the boundary of the slit disk, denoted by $\partial \Delta_{\epsilon}$, is a closed curve which is homotopic to the boundary of the approximating $\Delta$-domains.

\begin{corollary}
There exists an $\epsilon>0$ such that the generating function $(x,y)\mapsto\tilde{Z}(x,y)$ is holomorphic in $\Delta_\epsilon^2$ and continuous in its closure. 
\end{corollary}

\begin{proof}
The proof is a mutatis mutandis of the proof of \cite[Proposition 11]{CT21}, which uses the above lemmas as building blocks. The only remaining thing to check is that $\partial_H\check{Z}(1,1)\ne 0$, which is done by an explicit computation \cite{CAS3}.
\end{proof}

\paragraph{Step-by-step asymptotics.} Now the parametrization of the generating function $Z$ has the local expansion
\begin{equation*}
\hat Z(H,K)\ =\ \hat Z(H,1)	-\frac{\partial_K^3\hat Z(H,1)}6    (1-K)^3
							+\frac{\partial_K^4\hat Z(H,1)}{24} (1-K)^4 +O((1-K)^5)\,.
\end{equation*} Moreover, 
\begin{displaymath}
\hat{u}(K)=u_c-u_3(1-K)^3-u_4(1-K)^4+O\big((1-K)^5\big)
\end{displaymath}where $u_3$ and $u_4$ are some positive coefficients, giving
 $(1-K)^3 = \frac{u_c}{u_3} \mb({1-\frac{v}{u_c}} - \frac{u_4}{u_3}(1-K)^4+O\big((1-K)^5\big)$. In particular, we have $1-K \sim (\frac{u_c}{u_3})^{1/3} (1-\frac v{u_c})^{1/3}$ as $K\to 1$. Hence, we obtain the following expansion of $v\mapsto Z(u,v)$ at $v=u_c$:
\begin{equation}\label{eq:localZuv}
Z(u,v) = Z(u,u_c) - \partial_v Z(u,u_c) (u_c-v) + (A(u)-a_0)\mB({ 1-\frac v{u_c} }^{4/3}
											+ O\mb({ \mB({ 1-\frac v{u_c} }^{5/3} }\,,
\end{equation}
where $A(u)=\sum_{p\ge 0}a_p u^p$ is given by the rational parametrization $u=\hat u(H)$ and
\begin{equation*}
A(u)-a_0=\hat A(H)\ :=\	\mB({\frac{u_c}{u_3}}^{4/3} \m({ \frac{\partial_K^4\hat Z(H,1)}{24}
									  + \frac{\partial_K^3\hat Z(H,1)}{6}\cdot\frac{u_4}{u_3} }
	=  \frac{2^{2/3}H}{5H-10}\,.
\end{equation*}The function $\hat A$ has the expansion\begin{displaymath}
\hat{A}(H)=A_0+A_1(1-H)+O\big((1-H)^2\big)
\end{displaymath}which yields \begin{equation*}
A(u) = A(u_c) + b\mB({ 1-\frac{u}{u_c} }^{1/3}
		+ O\mb({ \mB({ 1-\frac{u}{u_c} }^{2/3} }
\end{equation*}
where $b=-(\frac{u_c}{u_3})^{1/3} A_1=-\frac{2}{5} 2^{1/3}$. 

It is easy to see by an explicit computation that similar expansions also hold for $Z_0$, with the coefficient of the dominant singularity term being $a_0=\frac{2^{2/3}(\sqrt{7}-1)(4-\sqrt{7})}{25\sqrt{7}-55}$. Thus like shown in \cite{CT20}, the transfer theorems of \cite{FS09} yield the asymptotics 
\begin{align}\label{eq:onestepasympt}
Z_q(u) & \eqv{q} \frac{A(u)-a_0}{\Gamma(-4/3)} u_c^{-q} q^{-7/3} \\
z_{p,q}	& \eqv{q} \frac{a_p}{\Gamma(-4/3)} u_c^{-q} q^{-7/3}	\notag\\
a_p		& \underset{p\to\infty}= \frac{b}{\Gamma(-1/3)} u_c^{-p} p^{-4/3} + O(p^{-5/3})\notag
\end{align}where the constants $a_p$ are given by $A(u)=\sum_{p\ge 0}a_p u^p$.

\begin{remark}
Based on Lemmas \ref{lem:ucritc} and \ref{lem:upolec} and vanishing of the first two partial derivatives of $\check{Z}$, \cite[Proposition 23]{CT21} provides a complete recipe for writing down the local asymptotic expansions of the generating function $Z$ around the dominant singularity. We will use this result directly in the next paragraph.
\end{remark}

\paragraph{Diagonal asymptotics.} We have already seen that $\partial_K\check Z(H,1)=\partial_K^2\check Z(H,1)=0$. Following the algorithm given in \cite[Section 4]{CT21}, this fact together with Lemmas \ref{lem:ucritc} and \ref{lem:upolec} yield the following asymptotic expansion for $\mathring{Z}(u,v)$:\begin{equation}\label{eq:diaglocalexpc}
\mathring{Z}(u_c x,u_c y)=Z_{\text{reg}}(x,y)+b\cdot Z_{\text{hom}}(1-x,1-y)+O\left(\max\{|1-x|,|1-y|\})^2\right)
\end{equation}where $(x,y)\to(1,1)$ for $(x,y)\in\overline{\Delta}_\epsilon^2$, the function $Z_{\text{reg}}$ is a regular part which does not contribute to the asymptotics, and $Z_{\text{hom}}$ is a homogeneous function of order $5/3$, meaning that $Z_{\text{hom}}(\lambda s,\lambda t)=\lambda^{5/3}Z_{\text{hom}}(s,t)$ for every $\lambda>0$ (see \cite[Remark 24]{CT21}). This follows from \cite[Lemma 20,  Proposition 23]{CT21} and their proofs. 

Now the proof of the diagonal asymptotics part of \cite[Theorem 2]{CT21} applies to the local expansion \eqref{eq:diaglocalexpc}, and we deduce \begin{equation}\label{eq:diagasymptc}
z_{p,q} \underset{p,q\to\infty}= \frac{b\cdot c\left(\frac{q}{p}\right)}{\Gamma(-4/3)\Gamma(-1/3)} u_c^{-(p+q)} p^{-11/3}+O\left(p^{-4}\right)\qquad\text{while}\quad \frac{q}{p}\in [\lambda_{\min},\lambda_{\max}]\subset(0,\infty)
\end{equation} where \begin{equation} 
c(\lambda) =\frac43 \int_0^\infty (1+r)^{-7/3}(\lambda+r)^{-7/3} \dd r.
\end{equation}During the course of the procedure, we check that the value of the constant $b$ indeed coincides with the one in the one-step asymptotics; this is done in the Maple worksheet \cite{CAS3}.

\subsection{Singularity analysis at $\nu<\nu_c$}

We conduct now the singularity analysis of the previous subsection for an arbitrary $\nu\in (1,\nu_c)$. Recall that for every $\nu\in (0,\nu_c)$, there is a unique $S\in (0,S_c)$ such that $\nu=\hat{\nu}(S)$. In addition to the physical high temperature regime, this range comprises the percolation $\nu=1$, corresponding to $S_{perc}:=\frac{1}{2}-\frac{\sqrt{3}}{6}$, as well as the antiferromagnetic regime $\nu\in (0,1)$. Applying again the \emph{parametrization} method of the \emph{algcurves}-package of Maple to equation \eqref{peq1cat3} with $\nu=\hat{\nu}(S)$ and $T=\hat{T}(S)$, we find a rational parametrization $(\hat{U}(R,S),\hat{F}(R,S))$ for $(U,F)$, where the coefficients are rather complicated. However, there is a unique $R_0(S)$ such that $\hat{U}(R_0(S),S)=\hat{F}(R_0(S),S)=0$, and the derivative of  $\hat{U}(R,S)$ has again three roots (see \cite{CAS3}). By guess and check, we may choose a suitable one for the parameter of the dominant singularity. Calling this parameter $R_c(S)$, we make the change of variable $R=R_0(S)-\frac{R_0(S)-R_c(S)}{S}\cdot H$. This yields a remarkably simpler parametrization, which will parametrize the generating functions at the origin by $H=0$ and has a critical point $H_c(S):=S$. Then, in the variable $H$, we obtain the parametrizations $\hat{u}(H,S):=T_c(S)^{-2/3}\hat{U}(H,S)$ and $\hat{Z}_0(H,S):=T_c(S)^{-1}\hat{U}(H,S)\hat{F}(H,S)$ for $(u,Z_0(u))$ as follows: 
\begin{align*}
&\hat{u}(H;S)=T_c(S)^{-2/3}\frac{2-3S}{36S^2(S-1)^2}H\frac{(3S-2)H^2+(2-4S)H-6S^3+10S^2-3S}{H-2S}
\end{align*}
\begin{align}
\hat{Z}_0(H;S)&=\frac{36 H}{S^2(3S-2)(6S^2-10S+3)}\notag\\ &\cdot \Big(\frac{\frac{5}{2}HS^5+\left(-H^2-\frac{15}{2}H+\frac{2}{3}\right)S^4+\left(-2H^3+\frac{20}{3}H^2+\frac{19}{4}H-\frac{2}{9}\right)S^3}{(H-2S)^2}\notag\\&+\frac{H\left(H^3+H^2-\frac{53}{9}H-\frac{5}{6}\right)S^2-\frac{4(H^2-\frac{2}{3}H-1)H^2S}{3}+\frac{4H^3(H-1)}{9}}{(H-2S)^2}\Big).\end{align}
Equation \eqref{eq:masterFEsol} gives then a parametrization of $Z(u,v)$ of the form
\begin{equation}
\hat{Z}(H,K;S)=\frac{N(H,K;S)}{S^2(3S-2)(6S^2-10S+3)(H-2S)^2(K-2S)^2(H+K-2S)}
\end{equation}where $N(H,K;S)$ is a polynomial in its three variables, too long to be written here (see the Maple worksheet \cite{CAS3} instead). The essential thing to note at this point is the fact that $\hat{Z}$ has singularities at $H=2S$, $K=2S$ and $H+K=2S$, except for some specific values of $S$ which cancel the denominator. Moreover, we find out that
\begin{displaymath}
\hat{Z}(H,H;S)=\frac{\tilde{N}(H;S)}{(H-2S)^2S^2(3S-2)(6S^2-10S+3)},
\end{displaymath}having only singularity at $H=2S$ (see an explicit expression in \cite{CAS3}). In addition, an explicit computation shows that this singularity is removable. We also define $\check{Z}(H,K;S)=\hat{Z}(H,K;S)+\hat{Z}_0(H;S)+\hat{Z}_0(K;S)$ which parametrizes $\mathring{Z}(u,v;\nu)$ together with $\hat{u}(H;S)$ and $\hat{u}(K;S)$.

We still need to show that $H_c(S)=S$ is indeed the parameter of the dominant singularity. It is known that the parameter is necessarily a critical point of smallest modulus. Since $S>0$, this already rules out the unique pole of $\hat{Z}_0(H)$, namely $H=2S$. At this point, we restrict ourself to the physical range $S\in (S_{perc},S_c)$. Then, the equation $\partial_H\hat{u}(H,S)=0$ has precisely three solutions, namely \begin{displaymath}
H=\frac{6S^2-2S-1\pm\sqrt{108S^4-192S^3+108S^2-20S+1}}{6S-4}
\end{displaymath}in addition to $H=S$. Now it is easy to verify that on the physical interval of $S$, all the aforementioned parameters are positive and greater than $S$. Moreover, the multiplicative terms which only depend on $S$ in the denominator of $\check{Z}(H,K;S)$ do not vanish if $S\in (S_{perc},S_c)$. Thus, $H=S$ is the simple zero which can be chosen for the parameter for the dominant singularity.  We denote \begin{displaymath}
u_c(S):=\hat{u}(H_c(S);S)=\hat{u}(S;S)=\frac{(18S^2-18S+4)2^{1/3}}{\left(\frac{(3S-2)^2(6S^2-10S+3)}{S(1-S)^3}\right)^{2/3}S(1-S)},
\end{displaymath}which is positive for $S\in (S_{perc},S_c)$. Using the parametrization $\nu=\hat{\nu}(S)$, we also denote $u_c(\nu)=u_c(S(\nu))$.

\begin{lemma}
The series $(u,v)\mapsto \mathring{Z}(u,v;\nu)=Z_0(u;\nu)+Z_0(v;\nu)+Z(u,v;\nu)$ is absolutely convergent if and only if $|u|\le u_c(\nu)$ and $|v|\le u_c(\nu)$.
\end{lemma}

\begin{proof}
A mutatis mutandis of the proof of Lemma \ref{lem:absconvcrit}.
 \end{proof}
 
We again make the change of variables $(x,y)=\left(\frac{u}{u_c(\nu)},\frac{v}{u_c(\nu)}\right)$ and define $\tilde{Z}(x,y;\nu):=\mathring{Z}(u_c x,u_c y;\nu)$. The corresponding RP is $x=\hat{x}(H;S):=u_c(S)^{-1}\hat{u}(H;S)$ and $\tilde{Z}(x,y;\nu)=\check{Z}(H,K;S)$. As in the critical case, $\hat{x}$ induces a conformal bijection from a neighborhood of $H=0$ to the unit disk $\mathbb{D}$, which extends continuously to the boundary of $\mathbb{D}$. Let $\mathcal{H}_0(S)$ be the component of the preimage of $\hat{x}^{-1}(\mathbb{D})$ containing the origin, and let $\overline{\mathcal{H}}_0(S)$ be its closure. Then the singularity structure of the generating functions is determined by the following results:

\begin{lemma}\label{lem:ucrith}
For $S\in (S_{perc},S_c)$, the value $H=S$ is the unique critical point of $\hat{x}$ in $\overline{\mathcal{H}}_0(S)$, being of multiplicity $1$. 
\end{lemma}

\begin{proof}
The zeros of $\hat{x}'(H)$ coincide with the zeros of $\hat{u}'(H)$, and the value $H=S$ was shown to be a simple zero, having the smallest absolute value; all the zeros are positive if $S\in (S_{perc},S_c)$.
\end{proof}

\begin{lemma}\label{lem:upoleh}
The value $(H,K)=(S,S)$ is the unique pole of $\check{Z}$ in $\overline{\mathcal{H}}_0(S)^2$.
\end{lemma}

\begin{proof}
From the expressions of $\hat{Z}(H,K;S)$ and $\hat{Z}_0(H;S)$, we see that the possible poles are located at $H=2S$, $K=2S$ or $H+K=2S$. Now $H=2S$ cannot belong to $\overline{\mathcal{H}}_0(S)$, since the domain $\overline{\mathcal{H}}_0(S)$ has the topology of the disk and is symmetric with respect to the $H$-axis. By symmetry, the same holds for $K=2S$. Finally, let us assume that $(H,K)\in\overline{\mathcal{H}}_0(S)^2$ such that $H+K=2S$. Since $\tilde{Z}(x,y;\nu)$ is absolutely convergent in $\overline{\mathbb{D}}$, then necessarily the numerator of $\check{Z}$ vanishes for such $(H,K)$. In addition, if $N(H,K;S)$ and $D(H,K;S)$ are the numerator and the denominator of $\check{Z}(H,K;S)$ such that they are coprime with respect to each other, \cite[Lemma 17]{CT21} tells that $\partial_H N\partial_K D-\partial_K N\partial_H D=0$. Now the pair of equations \begin{displaymath}
\begin{cases}
N(H,2S-H;S)=0\\
(\partial_H N\partial_K D-\partial_K N\partial_H D)(H,2S-H;S)=0
\end{cases}
\end{displaymath}has only solutions $H\in\{0,S,2S\}$, which is shown explicitly with a Maple computation \cite{CAS3}. Thus, $(H,K)=(S,S)$ is the unique pole.
\end{proof}

\begin{corollary}
For any $\nu\in(1,\nu_c)$ and $\theta\in\left(0,\frac{\pi}{2}\right)$, there exist an $\epsilon>0$ such that the generating function $(x,y)\mapsto\tilde{Z}(x,y;\nu)$ is holomorphic in $\Delta_\epsilon\times\Delta_{\epsilon,\theta}$ and continuous in its closure. 
\end{corollary}

\begin{proof}
Again, we check that $\partial_H\check{Z}(S,S;S)\ne 0$ by an explicit computation. The rest goes as in the proof of \cite[Proposition 11]{CT21}.
\end{proof}

\paragraph{Local expansions of $Z(u,v;\nu)$ and the diagonal asymptotics.} First, we check in \cite{CAS3} that $\partial_K\check Z(H,S;S)=0$ and $\partial_K^2\check Z(H,S;S)\ne 0$. Following again the algorithm given in \cite[Section 4]{CT21}, this together with Lemmas \ref{lem:ucrith} and \ref{lem:upoleh} yield the following asymptotic expansion for $\mathring{Z}(u,v;\nu)$:\begin{equation}\label{eq:diaglocalexpch}
\mathring{Z}(u_c(\nu) x,u_c(\nu) y)=Z_{\text{reg}}(x,y;\nu)+b(\nu)\cdot Z_{\text{hom}}(1-x,1-y;\nu)+O\left(\max\{|1-x|,|1-y|\}\right)
\end{equation}where $(x,y)\to(1,1)$ for $(x,y)\in\overline{\Delta_\epsilon\times\Delta_{\epsilon,\theta}}$, the function $Z_{\text{reg}}$ is a regular part which does not contribute to the asymptotics, and $Z_{\text{hom}}$ is a homogeneous function of order $1/2$. 

The proof of the diagonal asymptotics part of \cite[Theorem 2]{CT21} again applies to \eqref{eq:diaglocalexpch}, and we deduce \begin{equation}\label{eq:diagasympth}
z_{p,q} \underset{p,q\to\infty}= \frac{b(\nu)\cdot c\left(\frac{q}{p}\right)}{\Gamma(-3/2)} u_c(\nu)^{-(p+q)} p^{-5/2}+O\left(p^{-3}\right)\qquad\text{while}\quad \frac{q}{p}\in [\lambda_{\min},\lambda_{\max}]\subset(0,\infty)
\end{equation} where $c(\lambda) =(1+\lambda)^{-5/2}.$ During the course of the procedure (see the Maple worksheet \cite{CAS3}), we find the parametrization for $b(\nu)$ in the variable $S$ expressed as \begin{equation}
\hat{b}(S)=\frac{4(9S^2-10S+2)^2}{\sqrt{\frac{9S^2-10S+2}{3S^4-4S^3+S^2}}S^2(-18S^3+42S^2-29S+6)}.
\end{equation}We also check that $\hat{b}(S)$ is strictly positive when $S$ is in the physical interval.

\section{Peeling along the interface}\label{sec:peeling}

The basic idea for exploring the interface is to choose a \emph{local} and a \emph{Dobrushin-stable} peeling algorithm, i.e., an algorithm which chooses a boundary edge closest to the root whose deletion preserves the Dobrushin boundary and which, in the case of monochromatic boundary, chooses an edge whose endpoints have minimal graph distance to the origin. Such an exploration follows the interface in a natural way. In the previous works \cite{CT20} and \cite{CT21}, there was some freedom to choose such an algorithm, since the edge chosen by the algorithm could have two alternative boundary spins assigned. This freedom was also concretely exploited in \cite{CT21}, where the different behaviors of the associated perimeter processes led us to choose the starting point of the peeling at an edge with a prescribed spin for each temperature regime. In particular, the boundary edge there was always monochromatic. In this work, however, there is an obvious choice of the peeling algorithm: namely, the one which chooses the root edge itself. This is so since by definition, the root edge separates two opposite spins, i.e. is bichromatic. Moreover, such a peeling exploration always reveals a piece of unit length of the interface in one step. See Figure \ref{fig:interface} for graphical intuition.

\begin{figure}
\begin{center}
\includegraphics[scale=0.8]{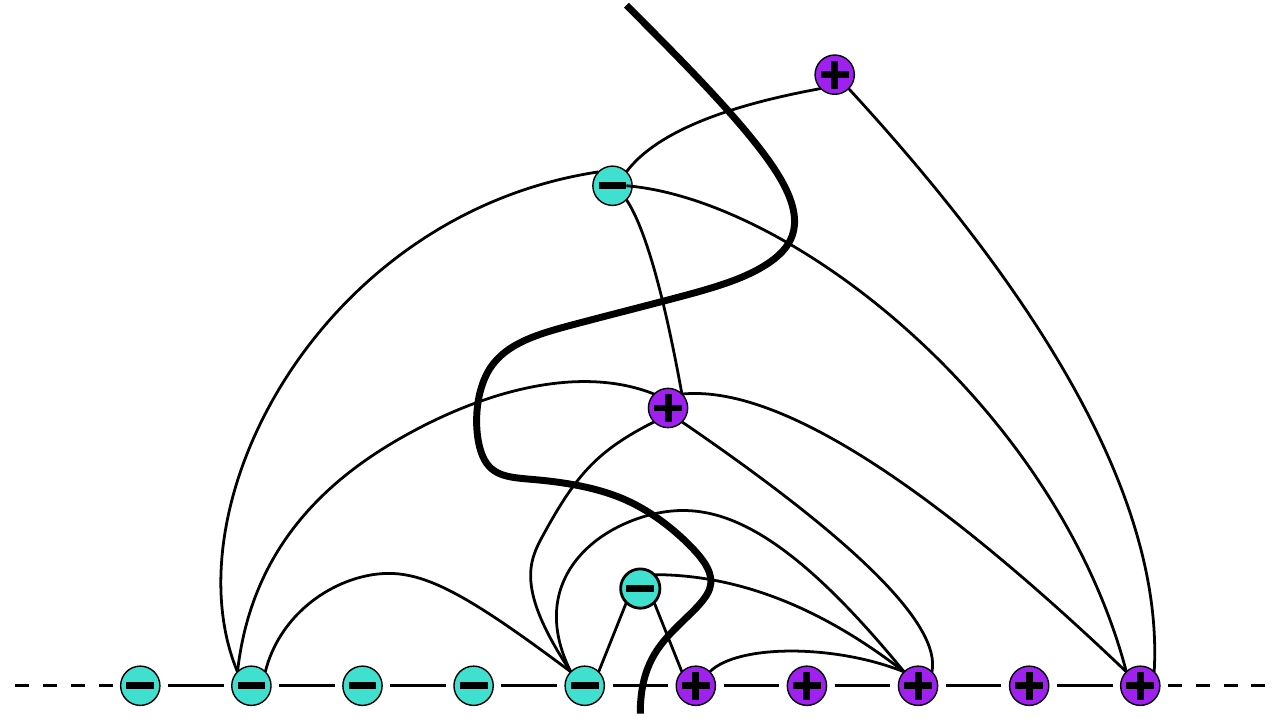}
\end{center}\vspace{-1em}
\caption{A portion of the interface explored by the peeling process. In this figure, the faces which do not contain a piece of the interface represent the swallowed region, and they have the law of a finite Boltzmann Ising-triangulation. This figure already demonstrates that the length of the interface in $n$ peeling steps is exactly $n$.}\label{fig:interface}
\end{figure}

\newcommand*{\hl}{\\\cline{1-6}}
\tabulinesep=1.2mm
\newcolumntype{L}{>{\!$\displaystyle}l<{$\!}}
\newcolumntype{S}{>{\!$\displaystyle}l<{$\!\!}}
\newcolumntype{R}{>{\!$\displaystyle}r<{$\!}}
\newcolumntype{C}{>{\!$\displaystyle}c<{$\!}}

\newcommand*{\zp}[1]{t\zz{#1}}
\newcommand*{\zzp}[2]{t\zzz{#1}{#2}}

\newcommand*{\zph}[1]{t\zzh{#1}}
\newcommand*{\zzph}[2]{t\zzzh{#1}{#2}}

\newcommand*{\zq}[1]{t\zz[p,q]{#1}}
\newcommand*{\zqh}[1]{t\zz[p,q]{#1}}
\newcommand*{\zqo}[1]{t\zz[0,q]{#1}}

\newcommand*{\zzq}[2]{t\zzz[p,q]{#1}{#2}}
\newcommand*{\zzqh}[2]{t\zzz[p+1,q]{#1}{#2}}
\newcommand*{\zzqo}[2]{t\zzz[0,q]{#1}{#2}}

\newcommand*{\ap}[1]{\,\frac{a_{#1}}{a_{p}}\,}
\newcommand*{\aph}[1]{\,\frac{a_{#1}}{a_{p+1}}\,}
\newcommand*{\apo}[1]{\,\frac{a_{#1}}{a_0}}

\newcommand{\pq}[1][P_n,Q_n]{\raisebox{-2pt}{$\!_{#1}\!$}}

\begin{table}[b!]
\begin{center}
\begin{tabu}{|L|S|S|  |L|S|S| L}
\cline{1-6}
\step &\Prob_{p,q}(\Step_1=\step) &(X_1,Y_1) &
\step &\Prob_{p,q}(\Step_1=\step) &(X_1,Y_1)
\hl \cp	&  \zph{p+1,q}				& (1,0)
&	\cm	&		\zph{p,q+1} 				& (0,1)
\hl \lp	&  \zzph{p,q-k}{0,k+1}		& (0,-k)
&	\rp[p+k]	&		\zzph{p,k+1}{0,q-k}		& (0,-q+k+1) &(0\le k\le q-1)
\hl \rp	&  	\zzph{p-k,q}{0,k+1}		&  (-k,0)
&	\lp[q+k]	&	\zzph{k+1,q}{p-k,0}	& (-p+k+1,0) &(0\le k\le p-1)
\hl
\end{tabu}
\vspace{3ex}

\begin{tabu}{|L|L|L| |L|L|L| L}
\cline{1-6}
 \step	& \Prob_{0,q}(\Step_1=\step)& (X_1,Y_1) & \step  & \Prob_{0,q}(\Step_1=\step) & (X_1,Y_1)
\hl \cp	& \nu\zqo{1,q}	& (1,0) &\cm		& \nu \zqo{0,q+1} & (0,1)
\hl \lp	& \nu\zzqo{0,q-k}{0,k+1} & (0,-k)	&\rp		& \nu\zzqo{0,q-k}{0,k+1} & (0,-k) &(0\le k\le q-1)
\hl
\end{tabu}
\caption{The probabilities of the first peeling event $\Step_1$ under $\Prob_{p,q}$ and $\Prob_{0,q}$ ($p,q\ge 1$). Here, we denote $t:=t_c(\nu)$ for simplicity. Observe that some symbols represent the same event. The associated perimeter variations $(X_1,Y_1)$ for the peeling with target $\rho^\dagger$ are also presented.}\label{tab:Ppq}
\end{center}\vspace{-0.7em}
\end{table}

\subsection{One-step peeling operation} Dividing Equation \eqref{receq1} by $z_{p,q}$ reveals a probability distribution. This distribution can be seen as the distribution of the first step in the peeling process as the root edge is deleted and its adjacent triangle is revealed. Formally, let $\steps:=\{\cp,\cm\}\cup\{\lp,\rp,: k\ge 0\}$ be a set of symbols. Assume that the vertex-bicolored triangulation $\bt$ has at least one boundary edge. Let $\algo$ be a peeling algorithm which chooses an edge $e$ from the boundary. We remove $e$ and reveal the internal face $f$ adjacent to it, together with the vertex $v$ at the corner of $f$ not adjacent to $e$. Observe that the vertex $v$ may still coincide with a vertex adjacent to $e$. If $f$ does not exist, then $\tmap$ is the edge map and $\bt$ has a weight 1 or $\nu$. Assume that $f$ exists. Then the peeling events are:
\begin{description}[noitemsep]
\item[Event $\CC^\+$:] $v$ is not on the boundary of $\tmap$ and has spin \+;
\item[Event $\CC^\<$:] $v$ is not on the boundary of $\tmap$ and has spin \<;
\item[Event $\RR_k$:] $v$ is at a distance $k$ to the right of $e$ on the boundary of $\tmap$ ($k\ge 0$);
\item[Event $\LL_k$:] $v$ is at a distance $k$ to the left of $e$ on the boundary of $\tmap$ ($k\ge 0$).
\end{description}

Let $\bt\in\mathcal{G}^\sigma_{p,q}$, and assume $p,q\ge 1$. In this case, the edge $e$ is chosen at the junction of the $\<$ and $\+$ edges, where the order of the edges is counterclockwise. At the very first step, $e=\rho$, the root edge. If $p=0$ and $q\ge 1$, the peeling algorithm chooses a monochromatic edge, and deletion of this edge gives a different law for the peeling due to the coefficient $\nu$. Then the Tutte's equations \eqref{receq1} and \eqref{receq2} define probability distributions, respectively, determined by the probabilities in Table \ref{tab:Ppq}; we denote these distributions by $\Prob_{p,q}$ and the random variable on $\steps$ by $\Step_1$ which takes a peeling step $\step\in\steps$ as a value. 

Now the diagonal asymptotics of equations \eqref{eq:diagasymptc} and \eqref{eq:diagasympth} give rise to a limit distribution as $p,q\to\infty$, denoted by $\Prob\yy$ and given in Table \ref{tab:Pinfty}. Moreover, in order to construct the local limit at $\nu=\nu_c$, we also need the limit of $\Prob_{p,q}$ as $p\ge 0$ is fixed and $q\to\infty$. This is denoted by $\Prob\py$, and its existence at $\nu=\nu_c$ follows from the one-step asymptotics \eqref{eq:onestepasympt}. This limit also exists for $1<\nu<\nu_c$, but we do not consider it here since it is not needed in the construction of the local limit in that case. That said, there is no problem to define the distribution $\Prob_p$ for arbitrary $\nu$. Table \ref{tab:Ppinfty} summarizes the corresponding transition probabilities for $p\ge 1$. If $p=0$, we need to take into account the monochromatic boundary, which gives the distribution as in Table \ref{tab:P0}.

\begin{table}[t]
\centering
\begin{tabular}{c r}
\begin{tabu}[t]{|L|L|L| |L|L|L| L}
\cline{1-6}
\step	& \Prob_\infty(\Step_1=\step) & (X_1,Y_1) & 
\step  & \Prob_\infty(\Step_1=\step) & (X_1,Y_1) &
\hl \cp	& \frac{t}{u} & (1,0)	&
	\cm		& \frac{t}{u} & (0,1) &
\hl \lp	&  \frac{t}{u} z_{k+1,0}u^{k+1}	& (0,-k) & 
	\rp	&  \frac{t}{u} z_{k+1,0}u^{k+1}  & ( -k,0) & (k\ge 0)
\hl
\end{tabu}
\end{tabular}
\caption{Law of $\Step_1$ under $\Prob_\infty$, obtained by taking the limit $p,q\to\infty$ in Table~\ref{tab:Ppq}, and the corresponding $(X_1,Y_1)$. We denote $t:=t_c(\nu)$ and $u:=u_c(\nu)$ for simplicity. 
}\label{tab:Pinfty}
\end{table}

\begin{table}[t]
\centering
\begin{tabular}{c r}
\begin{tabu}[t]{|L|L|L| |L|L|L| L}
\cline{1-6}
\step	& \Prob\py(\Step_1=\step)				& (X_1,Y_1)	&
\step	& \Prob\py(\Step_1=\step)				& (X_1,Y_1)	&
\hl	\cp & 		t \ap{p+1} 					& (1,0)	&
	\cm & \frac{t}{u}					& (0,1)	&
\hl	\lp & 		\frac{t}{u} z_{k+1,0} u^{k+1}		& (0,-k)	&
	\rp[p+k] & 		\frac{t}{u}\ap{0} z_{p,k+1} u^{k+1}	& (-p,-k)	&
(k\ge 0)
\hl	\rp & 		t \ap{p-k} z_{k+1,0} 	& (-k,0)	&
	 & 			& 	& (0\le k<p)
\hl
\end{tabu}
\end{tabular}
\caption{Law of $\Step_1$ under $\Prob\py$, $p\ge 1$, obtained by taking the limit $q\to\infty$ in Table~\ref{tab:Ppq}, and the corresponding $(X_1,Y_1)$ under the peeling without target.}\label{tab:Ppinfty}
\end{table}

\begin{table}[t]
\centering
\begin{tabular}{c r}
\begin{tabu}[t]{|L|L|L| |L|L|L| L}
\cline{1-6}
\step	& \Prob_0(\Step_1=\step) & (X_1,Y_1) & 
\step  & \Prob_0(\Step_1=\step) & (X_1,Y_1) &
\hl \cp	& t\nu\frac{a_1}{a_0} & (1,0)	&
	\cm		& \frac{t\nu}{u} & (0,1) &
\hl \lp	&  \frac{t\nu}{u} z_{k+1,0}u^{k+1}	& (0,-k) & 
	\rp	&  \frac{t\nu}{u} z_{k+1,0}u^{k+1}  & (0,-k) & (k\ge 0)
\hl
\end{tabu}
\end{tabular}
\caption{Law of $\Step_1$ under $\Prob_0$, obtained by taking the limit $q\to\infty$ with $p=0$ in Table~\ref{tab:Ppq}, and the corresponding $(X_1,Y_1)$. 
}\label{tab:P0}
\end{table}

\begin{lemma}
Let $p,q\ge 0$ be such that $p+q\ge 1$. Then for all $\nu\in (1,\nu_c]$,  $\Prob_{p,q}$,  $\Prob\py$ and $\Prob\yy$ are probability distributions on $\steps$.
\end{lemma}

\begin{proof}
The fact that $\Prob\yy$ is a probability distribution is straightforward to check: we have \begin{displaymath}
\sum_{\step\in\steps}\Prob\yy(\Step_1=\step)=2\frac{t_c(\nu)}{u_c(\nu)}\left(1+Z_0(u_c(\nu))\right)=1
\end{displaymath}by an explicit computation via Maple \cite{CAS3}, using the data of Table \ref{tab:Pinfty}. 

The proof that $\Prob\py$ defines a probability distribution is a bit more cumbersome, but has similar idea as in \cite{CT20}. First, we notice that $\sum_{\step\in\steps}\Prob_p(\Step_1=\step)=1$ for all $p\ge 1$ is equivalent to $\sum_{p\ge 1}a_pu^p=\sum_{p\ge 1}\sum_{\step\in\steps}a_p\Prob_p(\Step_1=\step)u^p$ as an equation of formal power series. The right hand side of the latter equation simplifies as \begin{displaymath}
t\left((1+Z_0(u_c))\frac{A(u)-a_0}{u_c}+(1+Z_0(u))\Delta_u A(u)+\frac{a_0}{u_c}(Z(u,u_c)-Z_0(u_c))-a_1\right)
\end{displaymath}where $\Delta_u A(u):=\frac{A(u)-a_0}{u}$. The left hand side is equal to $A(u)-a_0$, which is the coefficient of $\left(1-\frac{v}{u_v}\right)^{4/3}$ of the expansion of $Z(u,v)$ at $v=u_c$ as shown in Equation \eqref{eq:localZuv}. Expanding the first equation of \eqref{eq:masterFE} around $v=u_c$ finally justifies the above equations.

We still need to show that $\Prob_0$ defines a probability. For that purpose, we note that \begin{displaymath}
\sum_{\step\in\steps}\Prob_0(\Step_1=\step)=t_c(\nu)\nu\left(\frac{a_1(\nu)}{a_0(\nu)}+\frac{1}{u_c(\nu)}(1+2Z_0(u_c(\nu)))\right).
\end{displaymath} The expression on the right hand side is nothing but the coefficient of $\left(1-\frac{u}{u_c}\right)^{4/3}$ in the expansion of the right hand side of the second equation of \eqref{eq:masterFE} expanded around $u=u_c$, just divided by $a_0$. Hence, we are done.

\end{proof}

\subsection{Perimeter fluctuations} Applying the one-step peeling operation changes the length of the finite boundary. More precisely, after revealing $\Step_1=\step$ and erasing the edge $e$, and filling in a hole if the map is divided in two parts, the resulting map has both a different boundary condition and length. Let $(P_1,Q_1)$ be the boundary condition of the triangulation after this one-step peeling operation. If $p,q<\infty$, we define $(X_1,Y_1)=(P_1-p,Q_1-q)$. Observe that $(X_1,Y_1)$ actually only depends on the peeling step $\Step_1$ together with the rule that chooses the hole for the new unexplored part. Thus, $X_1$ and $Y_1$ can be seen as the relative changes of the positive and the negative boundary length, respectively, and hence can be defined also for $p=\infty$ or $q=\infty$. Tables \ref{tab:Pinfty}-\ref{tab:P0} summarize their values in the various cases of the peeling, where the rule for choosing the unexplored part depends on whether the peeling has a target or not. For more precise definitions, see Section~\ref{sec:sspeeling}.

It turns out that the expectation of $X_1$ (or equally $Y_1$) under $\Prob_\infty$ behaves like an order parameter. This is verified in the following lemma:

\begin{lemma}\label{lem:orderparam}
Under $\Prob_\infty$, the expectation of the perimeter variation is \begin{displaymath}
\EE_\infty(X_1)=\EE_\infty(Y_1)=\begin{cases}0\qquad\text{if}\quad 1<\nu<\nu_c\\
\frac{11-5\sqrt{7}}{12\sqrt{7}-48}=:\mu>0\qquad\text{if}\quad \nu=\nu_c
\end{cases}.
\end{displaymath}
\end{lemma} 

\begin{proof}
From Table \ref{tab:Pinfty}, we simply compute \begin{displaymath}
\EE\yy(X_1)=\frac{t_c(\nu)}{u_c(\nu)}-\sum_{k=1}^\infty k\cdot\frac{t_c(\nu)}{u_c(\nu)}z_{k+1}(\nu)u_c(\nu)^{k+1}=t_c(\nu)\left(\frac{1}{u_c(\nu)}(1+Z_0(u_c(\nu)))-Z_0'(u_c(\nu))\right).
\end{displaymath}This is done explicitly for $\nu\in(1,\nu_c)$ and $\nu=\nu_c$, respectively, in the Maple worksheet \cite{CAS3}. The claim for $Y_1$ follows by symmetry.
\end{proof}

\begin{remark}[Pure gravity]
Taking the limit $\nu\searrow 1$ yields $\Prob\yy(\Step_1=\cp)=\Prob\yy(\Step_1=\cm)=\frac{1}{2\sqrt{3}}$, which coincides with the corresponding probabilities for a peeling process on the type I UIHPT following a critical site percolation interface. Moreover, if $\mathcal{E}$ is the number of boundary edges swallowed by one peeling step, its expectation is \begin{displaymath}\EE_\infty(\mathcal{E})=2\sum_{k=1}^\infty k\cdot\frac{t_c(\nu)}{u_c(\nu)}z_{k+1}(\nu)u_c(\nu)^{k+1}=2\left(\frac{t_c(\nu)}{u_c(\nu)}-\EE_\infty(X_1)\right),
\end{displaymath}which tends to $2\frac{t_c(1)}{u_c(1)}=\frac{1}{\sqrt{3}}$ as $\nu\searrow 1$. This value corresponds to the expectation of the number of edges swallowed by an exploration of the critical site percolation on the UIHPT (see \cite{ACpercopeel}). Lemma \ref{lem:orderparam} essentially tells that the behavior of the peeling process in the high-temperature regime $1<\nu<\nu_c$ is similar to the behavior of the peeling process of the UIHPT.
\end{remark}

\subsection{The peeling process}\label{sec:sspeeling} For a given Ising-triangulation $\tmap$, the peeling is a deterministic exploration of a fixed map, driven by a \emph{peeling algorithm} $\mathcal{A}$. We assume the following: if the Ising-triangulation $\tmap$ has a bicolored boundary, the algorithm $\mathcal{A}$ chooses the edge at the junction of the $\<$ and $\+$ boundary segments on the boundary of the explored map, such that the starting point of the exploration is the root edge $\rho$. Since the deletion of that edge and the exposure of the adjacent face preserve the Dobrushin boundary condition, we say that $\mathcal{A}$ is \emph{Dobrushin-stable}. Otherwise, if the boundary of $\tmap$ is monochromatic, the algorithm $\mathcal{A}$ chooses the leftmost edge from the boundary with endpoints at minimal distance from the root, in the map explored thus far.

The peeling process along the interface $\iroot$ is constructed by iterating this face-revealing operation, yielding an increasing sequence $\nseq \emap$ of \emph{explored maps}. For each $\emap_n$, there is a unique unexplored map $\umap_n$, and a simple path of edges separating them from each other, which is called the \emph{frontier} and denoted by $\partial\emap_n$. Together, they compose the Ising-triangulation $\tmap$. More precisely, we set $\emap_0$ to be the boundary of $\tmap$, and $\emap_n$ is composed of $\emap_{n-1}$ and the revealed triangle $\Step_n$ together with a swallowed region. If there is no triangle to reveal, we set $\emap_n=\tmap$. The swallowed region is empty if $\Step_n\in\{\cp,\cm\}$. Otherwise, $\Step_n$ divides the unexplored part $\umap_{n-1}$ into two holes, and we choose $\umap_n$ according to the following rules: if the edge $\rho^\dagger$ is taken into account as a \emph{target}, then $\umap_n$ is the part containing $\rho^\dagger$, and the other region is filled. Otherwise, if the peeling is untargeted, we choose $\umap_n$ to be the region containing more $\<$ edges. In case of a tie, we choose the leftmost unexplored region from the root $\rho$. The root edge for each pair $(\emap_n,\umap_n)$ is the edge chosen by $\algo$, denoted by $\rho_n$. Since the setting is analogous to the one of spins on the faces, we refer to the works \cite{CT20} and \cite{CT21} for further details.

For practical purposes which will be clear later, we make the following choice for the peeling process: When considering the convergences $\prob_{p,q}^\nu\cv{q}\prob_p^\nu\cv{p}\prob\yy^\nu$, the peeling is chosen without target, which is well compatible with the infinite $\<$ boundary in particular when $\nu=\nu_c$. When $\prob_{p,q}^\nu\cv{p,q}\prob\yy^\nu$, we choose the peeling with the target $\rho^\dagger$, which ensures complete symmetry between the $\+$ and $\<$ spins.

At this point, we generalize the convergence of the one-step peeling of the previous subsection to the convergence of the whole peeling process as follows: First, we extend $\Prob_{p,q}$ as the law of $\nseq[1]\Step$ when $\bt$ is a Boltzmann-distributed Ising-triangulation of the $(p,q)$-gon. Then, it is enough to show that $\Prob_{p,q}(\Step_1=\step_1,\cdots,\Step_n=\step_n)$ converges for every $n\ge 1$ in any given regime of $(p,q)$, since the peeling process lives in a discrete space. This is done via the spatial Markov property of $\Prob_{p,q}$. We keep it informal, since it is already treated well in the works \cite{CT20}, \cite{CT21} and \cite{AMS18}, and just state the following:
\begin{proposition}\label{prob:peelingconv}
For all $\nu\in(1,\nu_c)$, the one-step peeling laws $\Prob\py$ and $\Prob\yy$ can be extended to laws of peeling processes on $\bts_\infty$ such that the following convergences holds weakly:\begin{align*}
& \Prob_{p,q}\cv[]q\Prob_p\cv[]p\Prob_\infty \\
& \Prob_{p,q}\cv[]{p,q}\Prob\yy\qquad\text{while}\quad\frac{q}{p}\in [\lambda_{\min},\lambda_{\max}].
\end{align*}
\end{proposition}

For $p,q<\infty$, the peeling process $\nseq{\Step}$ defines the associated perimeter processes $\nseq{P_n,Q}$ as the perimeters of the unexplored maps $\umap_n$, and we set $(X_n,Y_n):=(P_n-p,Q_n-q)$. Since $X_n$ and $Y_n$ actually only depend on the peeling steps $\Step_k$ up to time $n$, they can be extended to $q=\infty$ and $p=q=\infty$, respectively, describing the relative perimeter fluctuations in $n$ peeling steps. This is completely similar as in the works \cite{CT20}, \cite{CT21}. Defined this way, the peeling process satisfies the \emph{spatial Markov property}, and the processes $\nseq{P_n,Q}$ and $\nseq{X_n,Y}$ are Markovian, too:

\begin{corollary}
Under $\Prob_{p,q}$ and conditional on $(\Step_k)_{1\le k\le n}$, the sequence $(\Step_{n+k})_{k\ge 0}$ has the law $\Prob_{P_n,Q_n}$. In particular, $\nseq{P_n,Q}$ is a two-dimensional Markov chain.\\
Under $\Prob\py$ and conditional on $(\Step_k)_{1\le k\le n}$, the sequence $(\Step_{n+k})_{k\ge 0}$ has the law $\Prob\py[P_n]$. In particular, $\nseq P$ is a Markov chain.\\
Under $\Prob\yy$, the sequence $\nseq \Step$ is i.i.d. In particular, $\nseq{X_n,Y}$ is a two-dimensional random walk.
\end{corollary}

\section{Asymptotic properties of the perimeter processes}\label{sec:asympt}

By the asymptotics of $z_{p,q}$, it is easy to see that the distribution of $X_1$ and $Y_1$ under $\Prob\yy$ is heavy-tailed. More precisely, by equation \eqref{eq:onestepasympt} and its high-temperature counterpart, \begin{displaymath}
\Prob\yy(X_1=-k)=\Prob\yy(Y_1=-k)=\Prob\yy(\Step_1=\rp[k])=\frac{t_c(\nu)}{u_c(\nu)}z_{k+1,0}(\nu)u_c(\nu)^{k+1}\end{displaymath}\begin{displaymath}\underset{k\to\infty}\sim\begin{cases}\frac{t_c}{u_c}\frac{a_0}{\Gamma(-4/3)}k^{-7/3}\qquad (\nu=\nu_c)\\
\frac{t_c(\nu)}{u_c(\nu)}\frac{a_0(\nu)}{\Gamma(-3/2)}k^{-5/2}\qquad (1<\nu<\nu_c).\end{cases}
\end{displaymath}Above, the constant $a_0(\nu)$ can be computed from the asymptotic expansion of $Z_0(u;\nu)$; see the Maple worksheet \cite{CAS3} for the latter. It follows that $X_1$ and $Y_1$ belong to the domain of attraction of a totally asymmetric, spectrally negative stable distribution. Thus, the random walks $\nseq{X}$ and $\nseq{Y}$ have a scaling limit which is a stable Lévy process of index $4/3$ (when $\nu=\nu_c$) or $3/2$ (when $1<\nu<\nu_c$) with negative jumps. Moreover, although the random walks $\nseq{X}$ and $\nseq{Y}$ are not independent, they still have a joint scaling limit. This is proven in \cite{CurKPZ} and \cite{CT20} in similar settings, and the proofs extend to this case without an effort.

\begin{proposition}\label{prop:stable}
(1) For $\nu=\nu_c$, we have \begin{equation*}
\frac1{n^{3/4}}\mb({	 X_{\floor{nt}} -\mu nt,
					 Y_{\floor{nt}} -\mu nt }_{t\ge 0}
\cv[]n \mb({\mathcal X_t,\mathcal Y_t}_{t\ge 0} \,,
\end{equation*}
where $\mathcal X$ and $\mathcal Y$ are two independent and identically distributed spectrally negative $\frac43$-stable L\'evy processes of L\'evy measure $\frac{c_{\nu_c}}{\abs x^{7/3}}\idd{x<0}\dd x$, where $c_{\nu_c}$ is an explicit constant.
\\
(2) For $1<\nu<\nu_c$, we have \begin{equation*}
\frac1{n^{2/3}}\mb({	 X_{\floor{nt}} ,
					 Y_{\floor{nt}} }_{t\ge 0}
\cv[]n \mb({\tilde{\mathcal X}_t,\tilde{\mathcal Y}_t}_{t\ge 0} \,,
\end{equation*}
where $\tilde{\mathcal X}$ and $\tilde{\mathcal Y}$ are two i.i.d. spectrally negative $\frac32$-stable L\'evy processes of L\'evy measure $\frac{c_{\nu}}{\abs x^{5/2}}\idd{x<0}\dd x$, where $c_{\nu}$ is an explicit constant depending on $\nu$.
\\
Both of the above convergences take place in distribution w.r.t. the Skorokhod topology in the space of càdlàg functions.
\end{proposition}

\begin{remark}
In \cite{CT20}, when the spins are on the faces, the corresponding processes $\mathcal X$ and $\mathcal Y$ are not identically distributed, due to the fact the peeling process there is not symmetric with respect to the spins. 
\end{remark}
\newcommand{\tauxy}{\tau^\epsilon_x}
\newcommand{\barrier}[1][x]{#1 f_\epsilon}
Next, we move on to gather the asymptotic properties of the perimeter processes $P_n$ and $Q_n$, both under $\Prob\py$ and $\Prob_{p,q}$. What we are really interested in is the behavior of the hitting times of them in a neighborhood of the origin. Thus, for $m\ge 0$, define 
\begin{equation}
 T_m=\ \inf\Set{n\ge 0}{P_n\wedge Q_n\le m}.
\end{equation}Note that under $\Prob\py$, we have $Q_n=\infty$ almost surely. This definition makes sense if the peeling is with the target $\rho^\dagger$ when considering $\Prob_{p,q}$, and otherwise untargeted. We will see later that this hitting time corresponds approximately to the length of the main Ising interface imposed by the Dobrushin boundary conditions and followed by the peeling exploration. Since the interface behavior in the case of critical site percolation is already well-understood (\cite{ACpercopeel}), and our peeling process has a similar behavior at $\nu\in(1,\nu_c)$, we are mostly interested in the critical temperature $\nu=\nu_c$. There, the heavy-tailed distribution of $(X_1,Y_1)$ at the limit imposes a \emph{large jump} phenomenon, which was first discovered in \cite{CT20}, and extended to the diagonal limit $p,q\to\infty$ in \cite{CT21}. For that, fix $\epsilon>0$ and let \begin{equation*}
\barrier[](n) = \mb({ (n+2)(\log(n+2))^{1+\epsilon} }^{3/4}.
\end{equation*}Define the stopping time 

\begin{equation*}
\tauxy = \inf\Set{n\ge 0}{\abs{X_n-\mu n} \vee \abs{Y_n-\mu n} > \barrier(n) }\,.
\end{equation*}
where $x>0$.

\begin{lemma}[One jump to zero]\label{lem:one jump diag}
Assume $\nu=\nu_c$. Then for all $\epsilon>0$,
\begin{equation*}
\lim_{x,m \to\infty} \limsupp \Prob\py (\tauxy< T_m) = 0.
\end{equation*}
Moreover, for $0<\lambda_{\min}\leq 1\leq\lambda_{\max}<\infty$,
\begin{equation*}
\lim_{x,m \to\infty} \limsup_{p,q\to\infty} \Prob_{p,q} (\tauxy<T_m) = 0\quad\text{while}\quad\frac{q}{p}\in[\lambda_{\min},\lambda_{\max}].
\end{equation*}
\end{lemma}

The lemma says that the perimeter processes jump to a neighborhood of zero in a single big jump with high probability if $p$ and $q$ are large. This is a manifestation of the principle of a single big jump of heavy-tailed random walks, which is applied here to Markov chains with asymptotically heavy tails. Since the qualitative behavior of the perimeter processes here is similar to the behavior in \cite{CT20} and \cite{CT21}, the proof is a mutatis mutandis, and thus omitted. In fact, the proof is a bit simpler in this case, since the perimeter variation processes $\nseq{X}$ and $\nseq{Y}$ are identically distributed in the limit $p,q\to\infty$.

We give the following easy tail estimate for the distribution of $T_0$ at $\nu=\nu_c$ under $\Prob\py$. It is central in the proof of the local limit $\prob_{p,q}^{\nu_c}\cv{q}\prob_p^{\nu_c}$. The proof is similar as in \cite{CT20}.

\begin{lemma}[Tail of the law of $T_0$ under $\Prob\py$ at $\nu=\nu_c$]\label{lem:hit 0}
There exists $\gamma_0>0$ such that $\Prob\py(T_0> \Lambda p)\le \Lambda^{-\gamma_0}$ for all $p\ge 1$ and $\Lambda>0$. In particular, $T_0$ is finite $\Prob\py$-almost surely.
\end{lemma}

The above tail estimate can be generalized to the following scaling limit result, which actually comprises the main argument in the proof of Theorem \ref{thm:scaling}. 

\begin{prop}\label{prop:scalingT}
Let $\nu=\nu_c$. For all $m\in\natural$, the jump time $T_m$ has the following scaling limit:
\begin{equation*}
\forall t>0\,,\qquad	\lim_{p,q\to\infty} \Prob_{p,q}\m({\mu T_m>tp} = \int_t^\infty(1+s)^{-7/3}(\lambda+s)^{-7/3}ds
\end{equation*}where the limit is taken such that $q/p\to\lambda\in (0,\infty)$. In particular, for $\lambda=1$, \begin{equation*}
\lim_{p,q\to\infty} \Prob_{p,q}\m({T_m>tp} =(1+\mu t)^{-11/3}.
\end{equation*}
Moreover, 
\begin{equation*}
\lim_{p\to\infty} \Prob_p\m({T_m>tp} =(1+\mu t)^{-4/3}.
\end{equation*}
\end{prop}

\begin{proof}
The proof is mutatis mutandis of the proof for the similar claims in \cite{CT20} and \cite{CT21}. In particular, it uses Lemma~\ref{lem:one jump diag} as an input. We only need to take care of the correct exponents and the normalization by $\mu$, which are a priori not obvious from the asymptotics of $z_{p,q}$. We start from the latter claim, since it is simpler.

The core argument is the following: First, we notice that for large enough $p$, \begin{displaymath}
\Prob\py(T_m=1)=\Prob\py(P_1\le m)\sim \frac{\tilde c_m}{p}
\end{displaymath}for a constant $\tilde c_m$ depending on $m$, which can be explicitly computed:
\begin{align*}
\tilde c_m &=\lim_{p\to\infty}p\Prob_p(P_1\le m)=\lim_{p\to\infty}p\Prob_p(P_1=0)+\sum_{k=1}^m\lim_{p\to\infty}p\Prob_p(P_1=k) \\
&=\lim_{p\to\infty}p\sum_{k=0}^\infty \Prob_p(\Step_1=\rp[p+k])+\sum_{k=1}^m\lim_{p\to\infty} p\cdot\Prob_p(\Step_1=\rp[p-k]) \\
&=\frac{ta_0}{u_c}\left(\lim_{p\to\infty}p\frac{Z_p(u_c)}{a_p}+\frac{\Gamma\left(-1/3\right)}{\Gamma\left(-4/3\right)}b^{-1}\sum_{k=1}^m\lim_{p\to\infty} p\cdot\frac{(p-k+1)^{-7/3}}{p^{-4/3}}a_ku_c^k\right)
\\
&=\frac{ta_0}{b u_c}\frac{\Gamma\left(-1/3\right)}{\Gamma\left(-4/3\right)}\left(A(u_c)-a_0+\sum_{k=1}^ma_ku_c^k\right).
\end{align*}Taking the limit $m\to\infty$ defines \begin{equation*}\label{eq:cinfty}
\tilde c_\infty:=\lim_{m\to\infty}\tilde c_m=-\frac{4}{3}\frac{2ta_0}{bu_c}(A(u_c)-a_0).
\end{equation*} Now an explicit computation gives $\tilde c_\infty=\frac{4}{3}\mu$, where $\mu$ is defined in Lemma \ref{lem:orderparam}. The rest of the claim is already proven in \cite{CT20}. The proof is based on a similar result as Lemma \ref{lem:one jump diag}, and repeating the arguments of the proof of \cite[Proposition 11]{CT20} gives $\lim_{p\to\infty} \Prob_p\m({T_m>tp} =(1+\mu t)^{-\frac{\tilde c_\infty}{\mu}}.$

For the scaling limit of $T_m$ in the diagonal setting, the proof outline is given in the case of \cite[Theorem 6]{CT21}. The essential computation is the following:\begin{displaymath}
\lim_{p,q\to\infty} \Prob_{p,q}\m({T_m>tp} = \exp\left(-\int_0^t c_\infty\left(\frac{\lambda+\mu s}{1+\mu s}\right)\frac{ds}{1+\mu s}\right)
\end{displaymath} where \begin{align}\label{eq:cinfty}
c_\infty(\lambda)&:=\lim_{m\to\infty}\lim_{p,q\to\infty}\left(p\cdot\Prob_{p,q}(P_1\wedge Q_1\leq m)\right)=\lim_{m\to\infty}\lim_{p,q\to\infty}p\cdot\sum_{k=1}^m\left(\Prob_{p,q}(\rp[p-k])+\Prob_{p,q}(\lp[q-k])\right) \\
&=t\lim_{m\to\infty}\sum_{k=1}^m\lim_{p,q\to\infty}p\cdot\left(\frac{z_{k,q}z_{0,p-k+1}+z_{p,k}z_{0,q-k+1}}{z_{p,q}}\right)=2t\sum_{k=1}^\infty\frac{\Gamma\left(-1/3\right)}{\Gamma\left(-4/3\right)}\frac{a_0}{bu_c c(\lambda)\lambda^{7/3}}a_ku_c^k\notag \\
&=-\frac{4}{3}\frac{2ta_0}{bu_c c(\lambda)\lambda^{7/3}}(A(u_c)-a_0)=\frac{\tilde c_\infty}{c(\lambda)\lambda^{7/3}}=\frac{4}{3}\frac{\mu}{c(\lambda)\lambda^{7/3}}\notag.
\end{align}
\end{proof}

Moreover, we have the following bounds if we relax the assumption of the diagonal convergence to be as in Theorem \ref{thm:cv}:

\begin{prop}\label{prop:scalingT2}
For all $m\in\natural$, the scaling limit of the jump time $T_m$ has the following bounds:
\begin{equation*}
\forall t>0\,,\qquad	\liminf_{p,q\to\infty} \Prob_{p,q}\m({\mu T_m>tp} \ge \exp\left(-\int_0^t\max_{\lambda\in[\lambda_{\min},\lambda_{\max}]}c_\infty\left(\frac{\lambda+\mu s}{1+\mu s}\right)\cdot\frac{ds}{1+\mu s}\right)
\end{equation*} and \begin{displaymath}
\limsup_{p,q\to\infty} \Prob_{p,q}\m({\mu T_m>tp} \le \exp\left(-\int_0^t\min_{\lambda\in[\lambda_{\min},\lambda_{\max}]}c_\infty\left(\frac{\lambda+\mu s}{1+\mu s}\right)\cdot\frac{ds}{1+\mu s}\right)
\end{displaymath} where $c_\infty$ is the function defined by \eqref{eq:cinfty}, and the limit is taken such that $q/p\in [\lambda_{\min},\lambda_{\max}]$.
\end{prop}

\section{Local limits}\label{sec:locallimit}

In the proof of the local convergence of the laws $\prob_{p,q}^\nu$ and $\prob_p^{\nu_c}$, we use the following characterization of the local convergence: if $(\prob\0n)_{n\ge 0}$ and $\prob\0\infty$ are probability measures on $\overline{\bts}$, then $\prob\0n$ converges weakly to $\prob\0\infty$ for $d\1{loc}$ if and only if
\begin{equation*}
\prob \0n([\tmap,\sigma]_r=\bmap) \ \cv[]n\ \prob \0\infty([\tmap,\sigma]_r=\bmap)
\end{equation*}
for every $r\ge 0$ and every ball $\bmap$ of radius $r$. See \cite{CT20} and \cite{CT21} for more details.

In this paper, we consider the local limits in the two regimes $\nu\in (1,\nu_c)$ and $\nu=\nu_c$, respectively. These two regimes are expected to have a non-trivial behaviour of the interface. The former of the two is easier, so we begin with it.

\subsection{The local limit $\prob_{p,q}^\nu\to\prob\yy^\nu$ at $\nu\in(1,\nu_c)$} We sketch briefly the construction of the local limit $\prob_\infty^{\nu}$. In fact, it follows the \emph{general algorithm for constructing local limits} introduced in \cite{CT21}. The starting point is the convergence $\Prob_{p,q}\cv[]{p,q}\Prob_\infty$ of Proposition \ref{prob:peelingconv}. 

Let $\emapo_n$ be the map obtained by removing from $\emap_n$ all boundary edges adjacent to the hole. Considering the sequence of these maps, the number of the remaining boundary edges stays finite and only depends on $(\Step_k)_{k\le n}$. Then it follows that this convergence can be extended as follows \cite[Proposition 30]{CT21}: if $\theta$ is a $\Prob_\infty$-almost surely finite stopping time with respect to the filtration generated by the peeling process, then \begin{equation}\label{eq:stoppedconv}
\Prob_{p,q}(\emapo_\theta = \bmap) \cv[]{p,q} \Prob_\infty(\emapo_\theta = \bmap) \,.
\end{equation} This is all we need for the construction of the local limit $\prob\yy^\nu$. Namely, let $\theta=\theta_r=:\inf\Set{n\ge 0}{ d_{\emap_n}(\rho,\frontier_n)\ge r}$, where $d_{\emap_n}(\rho,\frontier_n)$ is the minimal graph distance in $\emap_n$ between $\rho$ and vertices on $\frontier_n$. Now \begin{equation*}
 \btsq_r\ =\ [\emapo_{\theta_r}]_r
\end{equation*}
for all $r\ge 0$. It follows that the peeling process $\nseq\emap$ eventually explores the entire triangulation $\bt$ if and only if $\theta_r<\infty$ for all $r\ge 0$. The latter follows, since the random walks $\nseq{X}$ and $\nseq{Y}$ have zero drift by Lemma \ref{lem:orderparam}. More precisely, it is well-known that one dimensional random walks on the real line with a zero drift are recurrent, and from this it follows that any finite segment of edges in the boundary either to the left or to the right of the root $\rho$ is swallowed by the peeling process in a finite time almost surely. 

Denote the law of the sequence of the explored maps under $\Prob\yy$ by $\law\yy\nseq{\emap}$. The local limit $\prob\yy^\nu$ is then defined as a growing sequence of finite balls $\law\yy \btsq_r := \lim\limits_{n \to\infty} \law\yy [\emap_n]_r$. The external face of $\law\yy \bt$ obviously has infinite degree and every finite subgraph of $\law\yy \bt$ is covered by $\emap_n$ almost surely for some $n<\infty$. Since the peeling process only fills in finite holes, it follows that the complement of a finite subgraph only has one infinite component. That is, $\prob\yy^\nu$ is one-ended, which together with the infinite boundary tells that the local limit is an infinite bicolored triangulation of the half-plane.

After all, the proof of the local convergence of $\prob_{p,q}^\nu$ towards $\prob\yy^\nu$ is just a one-line argument: since $\btsq_r = [\emapo_{\theta_r}]_r$ is a measurable function of $\emapo_{\theta_r}$, it follows from equation \eqref{eq:stoppedconv} that $\prob_{p,q}^\nu(\btsq_r = \bmap) \cv[]{p,q} \prob\yy^\nu(\btsq_r = \bmap)$ for every $r\ge 0$ and every ball $\bmap$. This implies the local convergence $\prob_{p,q}^\nu \cv{p,q} \prob\yy^\nu$.

Above, we did not actually need the information whether the peeling process takes into account the target $\rho^\dagger$ or not. In other words, the above construction gives the same result for both of the cases, and hence we have the freedom to choose.

\begin{remark}
The reason we did not consider the local convergence $\prob_{p,q}^\nu\cv{q}\prob_p^\nu$ is the fact that $\EE_p(X_1)$ and $\EE_p(Y_1)$ only have asymptotically zero drift when $p\to\infty$, which itself is not enough for their recurrence (compare to \cite{CT21} for an example of an asymptotically negative drift in the high-temperature regime). However, there are indeed some criteria to show the recurrence of such Markov chains, whose modifications could apply to the setting of this work. See \cite{GMPW19} and the references therein. We aim to go back to this question in future work.
\end{remark}

\subsection{The local limit $\prob_{p,q}^\nu\to\prob_p^\nu$ at $\nu=\nu_c$} This case is essentially similar to the previous one, except we replace $\prob\yy^{\nu}$ by $\prob_p\equiv\prob_p^{\nu_c}$ and the diagonal convergence by a univariate convergence. This does not change the proof much, since the only essential input is the convergence of the peeling process and the finiteness of $\theta_r$ under $\prob_p$. The latter one follows in this case from Lemma \ref{lem:hit 0}: since the boundary becomes monochromatic in a finite time almost surely under $\prob_p$, an analog of Lemma 29 in \cite{CT21} shows that indeed $\theta_r<\infty$. For more details, see \cite{CT20}.

\subsection{The local limits $\prob_p^\nu\to\prob\yy^\nu$ and $\prob_{p,q}^\nu\to\prob\yy^\nu$ at $\nu=\nu_c$}\label{sec:locallimitsc} First, we construct the local limit $\prob\yy\equiv\prob\yy^{\nu_c}$ using the positive drift of the perimeter processes $\nseq{X}$, $\nseq{Y}$ under $\Prob\yy$, the previously constructed local limit $\prob_0$ and a simple gluing argument. Then, we show the local convergence itself, which shares the same characteristics for both type of convergences.
\newcommand{\rib}{\emapo_\infty}
\paragraph{Construction of $\prob\yy$.} Due to the positive drift $\EE\yy(X_1)=\EE\yy(Y_1)>0$ (Lemma \ref{lem:orderparam}), the probability that the peeling process peels an edge adjacent to a given boundary vertex $v\in\partial\emapo_n$ infinitely many times is zero. Therefore, the sequence of balls $(\law\yy [\emapo_n]_r,\,n\ge 0)$ stabilizes in finite time for all $r\ge 0$. Hence, we may define $\law\yy{} [\rib]_r := \lim\limits_{n\to\infty} \law\yy{} [\emapo_n]_r$. We call it the \emph{ribbon}, for the reason that it is an infinite strip of triangles containing the infinite interface.

\begin{figure}
\begin{center}
\includegraphics[scale=0.8]{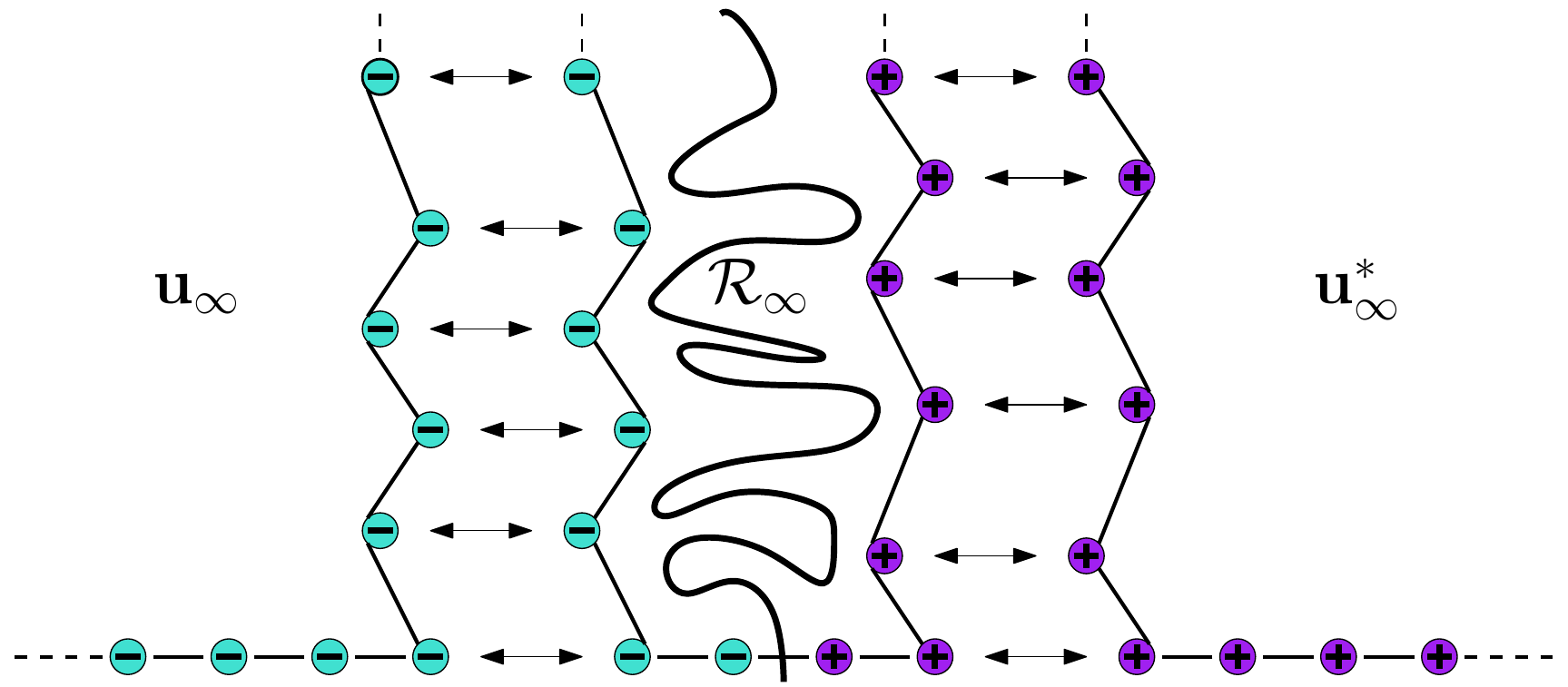}
\end{center}\vspace{-1em}
\caption{The gluing of the ribbon with the two infinite triangulations with a monochromatic boundary, in order to construct $\prob_\infty$. The vertices which are mutually identified are pointed with the double arrows.}\label{Fig:gluing}
\end{figure}

By construction, the ribbon is also one-ended: Namely, $\rib\setminus\emapo_n$ is connected, since the peeling process always reveals a triangle incident to the interface, and thus the consecutive revealed triangles share necessarily the edge which the interface traverses through. Moreover, the complement of any finite subgraph in $\rib$ has only one infinite connected component by the fact that it necessarily contains $\rib\setminus\emapo_n$ for $n\in\N$ sufficiently large. 

\newcommand*{\rmap}[1][m]{\mathcal{R}_{#1}}
\newcommand*{\uleft}[1][T_m]{\umap_{#1}}
\newcommand*{\uright}[1][T_m]{\umap^*_{#1}}
\newcommand{\kk}{\mathcal{K}_m}
\newcommand{\pjump}{\mathcal{P}}
\newcommand{\qjump}{\mathcal{Q}}
\newcommand{\pleft}{\mathcal{P}}
\newcommand{\qleft}{\mathcal{Q}}
\newcommand{\pright}{\mathcal{P}^*}
\newcommand{\qright}{\mathcal{Q}^*}

From now on, let us denote the ribbon under $\Prob\yy$ by $\rmap[\infty]$, and denote by $\overline \prob\py[0]$ the image of $\prob\py[0]$ in the inversion of spins. Let $\law\yy \uleft[\infty]$ and $\law\yy \uright[\infty]$ be two random variables of laws $\prob\py[0]$ and $\overline \prob\py[0]$, respectively, such that they are mutually independent with each other and $\law\yy \rmap[\infty]$. The boundary of $\law\yy \rmap[\infty]$ is partitioned into three intervals: one finite interval consisting of edges of $\emap_0$, and the two infinite intervals on its left and on its right. We glue $\law\yy \uleft[\infty]$ (resp.\ $\law\yy \uright[\infty]$) to the infinite interval on the left (resp. on the right), such that boundary vertices with the same spins are identified along the infinite boundaries, together with the incident edges. See Figure~\ref{Fig:gluing} for an explanation. Now $\prob\yy$ is defined as the law of the random triangulation resulted in this gluing. It is easy to see that $\prob\yy$ is one-ended, and that $\law\yy \nseq \emap$ is indeed the peeling process following the infinite interface of a random bicolored triangulation of law $\prob\yy$.

\paragraph{Convergence towards $\prob\yy$.} In order to show the local convergence, we want to find a counterpart of the above gluing argument for a triangulation $\tmap$ under $\Prob_p$ or $\Prob_{p,q}$ for some large $p,q$. There is no canonical way to do so, but instead we condition on the peeling step at time $T_m$ for some $m\ge 0$, and in the end take $m\to\infty$. In what follows, we formulate the proof primarily for $p,q<\infty$, and comment briefly what changes should take place when $q=\infty$. In the end, the detailed account of the latter is just a mutatis mutandis of the proof in the case of spins on the faces, found in \cite{CT20}.

To this end, fix $m\ge 0$, and define $\rmap[m]$ as the union of the explored map $\emapo_{T_m-1}$ and the triangle explored at $T_m$. Now the triple $(\umap_{T_m},\rmap[m],\umap_{T_m}^*)$ partitions a triangulation under $\prob_{p,q}$, such that $\umap_{T_m}$ and $\umap_{T_m}^*$ correspond to the two parts separated by the triangle at $T_m$. They correspond to the triangulations $\uleft[\infty]$ and $\uright[\infty]$ in the infinite setting, respectively. Observe that $T_m=\infty$ almost surely under $\prob\yy$, and hence this correspondence indeed makes formally sense. See Figure~\ref{fig:Decomposition}.

We will reroot the unexplored maps $\uleft$ and $\uright$ at the edges $\rho_\umap$ and $\rho_{\umap^*}$, which we define as the boundary edges with a vertex shared by $\umap_{T_m}$ and $\rmap[m]$, and $\umap_{T_m}^*$ and $\rmap[m]$, respectively. These edges are monochromatic, even though the boundary of $\uleft$ or $\uright$ might still be bichromatic. However, the triangulations $\uleft$ and $\uright$ look locally monochromatic with high probability when $p$ and $q$ are large. To formulate this, we import the following technical lemma introduced and proven in \cite{CT20}:

\begin{lemma}\label{lem:reroot}
Let $\prob_\pqq$ denote the pushforward of $\prob_{p,q_1+q_2}$ by the mapping that translates the origin $q_1$ edges to the left along the boundary. Then for all fixed $p\ge 0$, we have $\prob_\pqq \xrightarrow{d\1{loc}} \prob_0$ weakly as $q_1,q_2\to\infty$.
\end{lemma}

Now the boundary condition of $\uleft$ can be written as $(\pleft,(\qleft_1,\qleft_2))$ according to the notation of the previous lemma, where $\pleft$, $\qleft_1$ and $\qleft_2$ are some random numbers. Similarly, the boundary condition of $\uright$ can be written as $((\pright_1,\pright_2),\qright)$, where the notation is understood such that the two components of the pair are switched in a spin-flip, which produces the setting of Lemma~\ref{lem:reroot}. Let $\delta_\star$ be a random variable assigning value $1$ if the boundary vertex of the revealed triangle $\Step_{T_m}$ is of spin $\star$, and $0$ otherwise. Then, the condition 
$
\Step_{T_m} = \rp[P_{T_m-1}+\kk-\delta_{\<}]
$
uniquely defines an integer $\kk$, which represents the position relative to $\rho^\dagger$ of the vertex where the triangle revealed at time $T_m$ touches the boundary. We also make the convention $\rp[p+k]=\lp[q-k-1]$. See Figure \ref{fig:Decomposition}. 

\begin{figure}[H]
\begin{center}
\includegraphics[scale=0.7]{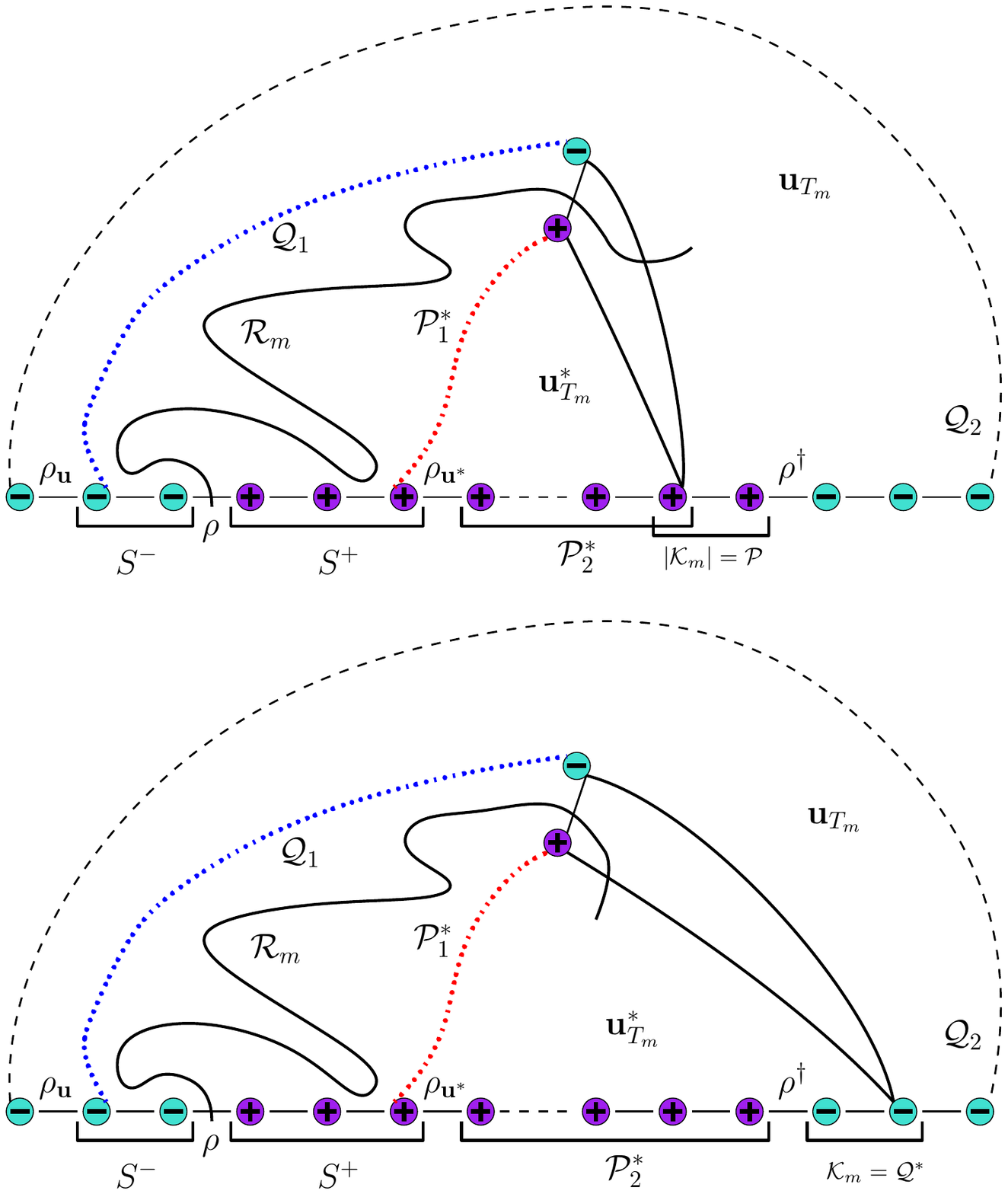}
\end{center}\vspace{-1em}
\caption{The decomposition of the Ising triangulation at time $T_m$, showing the two possible scenarios. This figure also shows how the interface behaves differently depending on the last peeling step.}\label{fig:Decomposition}
\end{figure}   

In the following lemma, we show that the ribbon and the unexplored parts converge jointly. It gathers analogous results from \cite{CT20} and \cite{CT21} with minor modifications to the setting of this work.
\newcommand{\PrE}[4]{\prob#1 \mb({ ([#2]_r,[#3]_r,[#4]_r) \in \mathcal{E} }}

\begin{lemma}[Joint convergence before gluing]\label{lem:loc cv on big jump}
Fix $\epsilon,x,m>0$, and let $\mathcal{J} \equiv \mathcal{J}^\epsilon_{x,m} := \{\tauxy = T_m \ge \epsilon p\}$. Then for any $r\ge 0$,
\begin{equation}\label{eq:loc cv on big jump}
\begin{aligned}
& \limsup_{p,q\to\infty} \abs{ \PrE\pqy{\rmap}{\uleft}{\uright}
               - \PrE\yy{\rmap[\infty]}{\uleft[\infty]}{\uright[\infty]} }
\\   \le \ &
\limsup_{p,q\to\infty} \prob_{p,q} (\mathcal{J}^c) + \prob_\infty(\tauxy<\infty)
\end{aligned}
\end{equation}
in the sense of the diagonal limit of Theorem~\ref{thm:cv} where $\mathcal{E}$ is any set of triples of balls.

In the case of $q=\infty$,  $\mathcal{J}$ has to be replaced by $\{\tauxy = T_m \ge \epsilon p\} \cap \{ \kk \le m\}$.
\end{lemma}

\begin{proof}
The proof is a mutatis mutandis of the proofs of \cite[Lemma 14]{CT20} and \cite[Lemma 41]{CT21}, where it is presented for the spins on the faces. Given the proof framework detailed in \cite{CT20}, the only thing one needs to take care of is the fact that the random numbers $\pright_1$, $\pright_2$, $\qleft_1$ and $\qleft_2$ tend to $\infty$ uniformly, and that $\pleft$ and $\qright$ stay bounded, conditional on $\mathcal{J}$. Observe also that the random number $\kk$ is automatically bounded in the diagonal setting, so we do not need any condition for $\kk$ in the event $\mathcal{J}$ in that case, whereas in the case $q=\infty$, $\kk$ is not bounded a priori.

Let us start by showing lower bounds for the boundary condition of $\uleft$. First, expressing the total perimeter of $\uleft$, the number of edges between $\rho$ and $\rho^\dagger$ clockwise and the number of \+ vertices on the boundary of $\uleft$, respectively, we find the equations \begin{displaymath}
\begin{cases}
\qleft_1+\qleft_2+\pleft=Q_{T_m-1}-\kk+\delta_{\<} \\
S^{\<}+\qleft_2+\max\{0,\kk\}-\delta_{\<}=q \\
\pleft=-\min\{0,\kk\}
\end{cases}
\end{displaymath} where $S^{\<}$ is the number of vertices of spin $\<$ in $\rmap[m]\cap\partial\emapo_0$. See Figure \ref{fig:Decomposition}. The solution of this system of equations is \begin{displaymath}
\begin{cases}
\qleft_1=Y_{T_m-1}+S^{\<} \\
\qleft_2=q-S^{\<}-\max\{0,\kk\}+\delta_{\<} \\
\pleft=-\min\{0,\kk\}
\end{cases}. 
\end{displaymath}We have $S^{\<}=-\min_{n<T_m} Y_n\in [0,1-\min_{n\ge 0}(\mu n-xf_\epsilon(n))]$, and the function $n\mapsto\mu n-x f_\epsilon(n)$ is increasing if we only consider large enough $n$. Therefore, conditional on $T_m\ge\epsilon p$, we deduce $\qleft_1\ge Y_{T_m-1}\ge\mu(T_m-1)-xf_\epsilon(T_m-1)\ge\mu(\epsilon p-1)-x f_\epsilon(\epsilon p-1)=:\underline{\qleft_1}$ and $\qleft_2\ge q+\min_{n\ge 0}(\mu n-x f_\epsilon(n))-1-m=:\underline{\qleft_2}$ for large enough $p$. Moreover, $\pleft\le|\kk|\le m$. 

By symmetry, for $\uright$ we have  
\begin{displaymath}
\begin{cases}
\pright_1+\pright_2+\qright=P_{T_m-1}+\kk+\delta_{\+} \\
S^{\+}+\pright_2-\min\{0,\kk\}-\delta_{\+}=p \\
\qright=\max\{0,\kk\}
\end{cases}
\end{displaymath} where $S^{\+}$ is now the number of vertices of spin $\+$ in $\rmap[m]\cap\partial\emapo_0$. It is easy to see that this yields similar bounds for $\pright_1$, $\pright_2$ and $\qright$ as previously for $\qleft_1$, $\qleft_2$ and $\pleft$, respectively. The rest of the proof goes then as in \cite{CT20}. Observe that if $q=\infty$, we need the condition $\kk\le m$ to ensure the boundedness of $\qright$, and we also have $\qleft_2=\infty$; otherwise the proof is the same.
\end{proof}

\newcommand{\pglued}[1]{\prob#1 \mb({ [\mop]_r \in \mathcal{E} }}
\newcommand{\mop}{\tmap \oplus \tmap'}

\paragraph{Proof of the convergence $\prob\pqy^{\nu_c}\to\prob\yy^{\nu_c}$.} The rest of the proof is mostly presented in detail in \cite{CT20} and \cite{CT21}. It is based on a gluing argument of three locally converging maps, which results the local convergence of the glued map itself. There, the gluing happens such that the spins assigned to the boundary edges on each of the side of the gluing interface (which is not to be confused with the Ising interfaces) coincide. In this article, we simply switch the roles of the boundary edges and the boundary vertices, and otherwise apply the same arguments. This is possible since the boundaries of the maps are simple, and thus there is a one-to-one correspondence between the boundary vertices and the boundary edges. That said, the to-be-glued boundary vertices on each side always have the same spin, and only monochromatic edges are merged in the gluing. In the next few paragraphs, we outline the existing arguments in the setting of this article.

Similarly as the infinite triangulation $\law\yy \bt$ can be represented as a gluing of the triple $\law\yy (\rmap[\infty],\uleft[\infty],\uright[\infty])$, the finite triangulation $\law\pqy \bt$ results from the gluing of the triple $\law\pqy (\rmap,\uleft,\uright)$ along the boundaries of the components. This is done pairwise between the three components, taking into account that the location of the root edge changes during this procedure. Given a triangulation $\tmap$ with a simple boundary, and an integer $S$, let us denote by $\overrightarrow \tmap^{S}$ (resp. $\overleftarrow \tmap^{S}$) the map obtained by translating the root edge of $\tmap$ by a distance $S$ to the right (resp. to the left) along the boundary. Denote by $\rho$ and $\rho'$ the root edges of two triangulations $\tmap$ and $\tmap'$, respectively, and let $L$ be the number of vertices in $\tmap$ and $\tmap'$ which are admissible for the gluing. More precisely, we assume that $L$ is a random variable taking positive integer or infinite values, such that
\begin{equation}\label{eq:gluing length}
\law\pqy L \cv[]{p,q} \infty \text{ in distribution and }\law\yy L = \infty \text{ almost surely.}
\end{equation} Finally, let $\mop$ be the triangulation obtained by gluing the $L$ boundary vertices of $\tmap$ on the right of $\rho$ to the $L$ boundary vertices of $\tmap'$ on the left of $\rho'$, together with the edges which have two such vertices as endpoints. The dependence on $L$ is omitted from this notation because the local limit of $\mop$ is not affected by the precise value of $L$, provided that \eqref{eq:gluing length} holds.

Now using the notation of the previous paragraph, we have 
\begin{equation}\label{eq:gluing-of-3}
\bt = \overrightarrow{(\umap\rmap)}^{S^\+ + S^\<} \oplus \uright \qtq{where} \umap\rmap = \uleft \oplus \overleftarrow{(\rmap)}^{S^\<}
\end{equation}
where $S^\+$ and $S^\<$ are the number of vertices between $\rho$ and $\rho_{\umap^*}$ or $\rho_\umap$, respectively, including the ones adjacent to the above root edges. Similarly, $\law\yy \bt$ can be expressed in terms of $\uleft[\infty]$, $\rmap[\infty]$, $\uright[\infty]$ and $S^\jj$ using the above described gluing and root translation.

On the event $\mathcal{J}$, the perimeter processes $\nseq X$ and $\nseq Y$ stay above $\mu n-xf_\epsilon(n)$ up to time $\tauxy$. Thus their minima over $[0,\tauxy)$ are reached before the deterministic time $N_{\min} = \sup\Set{n\ge 0}{\mu n-xf_\epsilon(n)\le 0}$ and $S^\+$ and $S^\<$ are measurable functions of the explored map $\emapo_{N_{\min}}$. It follows that $\law\pqy S^\jj$ converges in distribution to $\law\yy S^\jj$ on the event $\mathcal{J}$. Using the relation \eqref{eq:gluing-of-3} together with \cite[Lemmas 15-16]{CT20}, we deduce from Lemma~\ref{lem:loc cv on big jump} that for any $x,m,\epsilon>0$, and for any $r\ge 0$ and any set $\mathcal{E}$ of balls, we have
\begin{equation*}
\limsuppq \mb|{\, \prob\pqy ( \btsq_r \in \mathcal{E} ) 
               - \prob\yy( \btsq_r \in \mathcal{E} ) \, }    \ \le\ 
\limsuppq \prob\pqy (\mathcal{J}^c) + \prob\yy (\tauxy<\infty)    \,.
\end{equation*}

The left hand side does not depend on the parameters $x,m$ and $\epsilon$. Therefore to conclude that $\prob\pqy$ converges locally to $\prob\yy$, it suffices to prove that
$\displaystyle \limsuppq \prob\pqy (\mathcal{J}^c) + \prob\yy (\tauxy<\infty)$\, converges to zero when $x,m\to\infty$ and $\epsilon\to 0$. The latter term converges to zero, since if $x\to\infty$, we have $\tauxy \to\infty$ almost surely under $\prob\yy$. For the first term, a union bound gives
\begin{equation*}
\prob\pqy(\mathcal{J}^c) \ \le\ \prob\pqy (\tauxy< T_m) + \prob\pqy(T_m<\epsilon p) \,,
\end{equation*}
where the first term on the right can be bounded using Lemma~\ref{lem:one jump diag}:
\begin{equation*}
\lim_{m,x \to\infty}  \limsuppq \prob\pqy(\tauxy<T_m)    \ =\ 0 .
\end{equation*} For the last term, we use Proposition~\ref{prop:scalingT2}:
\begin{equation*}
\lim_{\epsilon\to 0}\  \limsuppq \prob\pqy(T_m<\epsilon p)\ \leq\ 
1-\lim_{\epsilon\to 0}\exp\left(-\int_0^\epsilon \max_{\lambda\in[\lambda_{\min},\lambda_{\max}]}c_\infty\left(\frac{\lambda+\mu s}{1+\mu s}\right)\frac{ds}{1+\mu s}\right) \ =\ 0    \,. 
\end{equation*}

In the case $q=\infty$, we have \begin{equation*}
\prob\py(\mathcal{J}^c) \ \le\ \prob\py (\tauxy< T_m) + \prob\py(T_m<\epsilon p) +\prob\py(\tauxy=T_m,\ \kk>m)\,,
\end{equation*}where the first term on the right hand side is treated like before and the second term is shown to be negligible by Theorem \ref{thm:scaling}. For the last term, we repeat an estimate in \cite{CT20} in order to find the bound \begin{align}\label{eq:limsupbound}
\limsup_{p\to\infty}\prob\py(\tauxy=T_m,\ \kk>m)&\leq\limsup_{p\to\infty}\Prob_p(\Step_1\in\{\rp[p+k-1]:\ k>m\}|P_1\le m)\notag\\&=\limsup_{p\to\infty}\frac{\sum_{k\ge m}p\cdot\Prob\py(\Step_1=\rp[p+k])}{\sum_{k\ge -m}p\cdot\Prob\py(\Step_1=\rp[p+k])}\,.
\end{align}By \eqref{eq:onestepasympt}, we see that \begin{equation}
\limsup_{p\to\infty} p\cdot\sum_{k\ge 0}\prob_p(\Step_1=\rp[p+k])=\frac{t_ca_0}{u_c}\lim_{p\to\infty}p\cdot\frac{Z_p(u_c)}{a_p}=-\frac{4}{3}\frac{t_ca_0}{b u_c}(A(u_c)-a_0).
\end{equation}Moreover, if $k\ge 0$, we have \begin{equation}
\lim_{p\to\infty}p\cdot\Prob\py(\Step_1=\rp[p+k])=\frac{t_c a_0}{u_c}\lim_{p\to\infty}p\cdot\frac{z_{p,k+1}}{a_p}u_c^{k+1}=-\frac{4}{3}\frac{t_ca_0}{b u_c}a_{k+1}u_c^{k+1}.
\end{equation}It follows that \begin{displaymath}
\limsup_{p\to\infty} p\cdot\sum_{k\ge 0}\prob_p(\Step_1=\rp[p+k])=\sum_{k\ge 0}\lim_{p\to\infty} p\cdot\prob_p(\Step_1=\rp[p+k]).
\end{displaymath}Moreover, if $k<0$, \begin{displaymath}
\lim_{p\to\infty} p\cdot\prob_p(\Step_1=\rp[p+k])=-\frac{4}{3}\frac{t_c a_0}{b u_c}a_{|k|}u_c^{|k|}<\infty.
\end{displaymath} It follows that the right hand side of \eqref{eq:limsupbound} converges to zero as $m\to\infty$. This finally proves the claim. \qed

\section{Scaling limits of the interface length}\label{sec:scalingl}

In this final section, we finish proving Theorem \ref{thm:scaling}. The proof relies on the observation than if $\tmap$ is sampled from $\prob_{p,q}$ or $\prob_p$ at $\nu=\nu_c$, the length of the main interface is close to the hitting time $T_m$ when $p$ and $q$ are large. If the spins were on the faces, the discrete interface would not be a simple curve, which prevented us from deducing a similar claim in \cite{CT20} and \cite{CT21}. Moreover, what is known about the qualitative behavior of triangulations of the half-plane decorated with the critical percolation \cite{ACpercopeel}, we deduce a scaling limit of the perimeter of the hull containing the portion of the interface before its first visit to the boundary of the half-plane when $\nu\in (1,\nu_c)$. This claim extends the result of Angel and Curien at $\nu=1$.
\\
\\
\emph{Proof of Theorem \ref{thm:scaling}.} We have almost proven the claim in Proposition \ref{prop:scalingT}. The rest of the proof resembles an argument used in \cite[Theorem 6]{CT20} to show that the scaling limit of $T_m$ is independent of $m$, which was later generalized in \cite{CT21}. In the case $q=\infty$, the idea is the following: if $T_m<T_0$ for some $m\ge 1$, we can decompose the interface length under $\prob_p$ as $\eta_p=T_m+\eta_{-\min(\kk,0)}$, where we recall that $\kk$ is the position relative to $\rho^\dagger$ of the vertex where the triangle revealed at time $T_m$ hits the boundary. The same idea generalizes to $\eta_{p,q}$ under $\prob_{p,q}$ as well. The reason why this argument works here is the fact that before the time $T_m$, each peeling step increases the interface length exactly by one, and after that hitting time, the length of the unexplored portion of the interface stays small compared to $p$ when $p\to\infty$. The former of the two does not hold when the spins are put in the faces (see \cite[Section 6]{CT20}).

The above idea in the case of $\eta_p$ actually requires some special care, since it is possible that $T_m=T_0$. Formally, if $\delta>0$, we estimate \begin{align}\label{eq:decomp}
&\Prob_p(\eta_p-T_m>\delta p)\notag\\ &\le \Prob_p(\eta_p-T_m>\delta p,\ T_m=\tauxy\ne T_0)+\Prob_p(\eta_p-T_m>\delta p,\ T_m=\tauxy= T_0)+\Prob_p(T_m<\tauxy).
\end{align}The third term of \eqref{eq:decomp} can be made arbitrarily small when $p,m,x$ are large thanks to Lemma~\ref{lem:one jump diag}. The first term can be estimated by strong Markov property:
 \begin{equation}\label{eq:scalingintp}
\Prob_p(\eta_p-T_m>\delta p,\ T_m=\tauxy\ne T_0)) = \EE\py\m[{ \Prob\py[P_{T_m}](\eta_{P_{T_m}}>\delta p)\id_{\{T_m=\tauxy\ne T_0\}} }
					 \le\	\max_{p'\le m}\Prob\py[p'](\eta_{p'} >\delta p)\,.
\end{equation}Now recall that by Lemma \ref{lem:hit 0}, the \+ boundary of length $p'$ is swallowed by the peeling almost surely in a finite time. The swallowed region is a finite Boltzmann Ising-triangulation, which includes the interface component of length $\eta_{p'}$. Thus, $\eta_{p'}<\infty$ almost surely for all $0\le p'\le m$, and therefore we conclude by \eqref{eq:scalingintp} that the first term of \eqref{eq:decomp} converges to zero as $p\to\infty$. For the second term, the trick is to run the peeling under the inversion of spins on $\umap_{T_m}^*$ as seen in the lower part of Figure~\ref{fig:Decomposition}. Using the notation of Section~\ref{sec:locallimitsc}, we estimate \begin{align}\label{eq:m0}
&\Prob_p(\eta_p-T_m>\delta p,\ T_m=\tauxy= T_0)\notag\\ &\le\sum_{k=1}^{k_0}\Prob_p(\eta_p-T_m>\delta p,\ T_m=\tauxy= T_0,\ \mathcal{K}_0=k)+\Prob_p(\mathcal{K}_0>k_0,\ T_m=\tauxy=T_0)\notag\\ &\le \sum_{k=1}^{k_0}\EE\py\m[{\Prob_{P_{T_0-1},k}(\eta_{P_{T_0-1},k}>\delta p)\id_{\{T_m=\tauxy=T_0\}}}+\EE\py\m[{\Prob_{P_{T_0}-1}(\Step_1\in\rp[P_{{T_0}-1}+k]:\ k\ge k_0)\id_{\{T_m=\tauxy=T_0\}}}.
\end{align} Conditional on $\{T_m=\tauxy=T_0\}$, we have $P_{T_0-1}\ge p+\mu(T_0-1)-xf_\epsilon(T_0-1)$, where the right hand side tends to infinity as $p\to\infty$. Thus, the second term on the right hand side of \eqref{eq:m0} is arbitrarily small if $k_0$ and $p$ are chosen to be large enough. Then, the first term on the right hand side of \eqref{eq:m0} tends to zero as $p\to\infty$. To put things together in equation \eqref{eq:decomp}, since the scaling limit of $T_m$ is independent of $m$, we deduce $\Prob_p(\eta_p-T_m>\delta p)\cv[]p 0$, which shows that $\eta_p$ and $T_m$ have the same scaling limit.

\begin{figure}
\begin{center}
\includegraphics[scale=0.8]{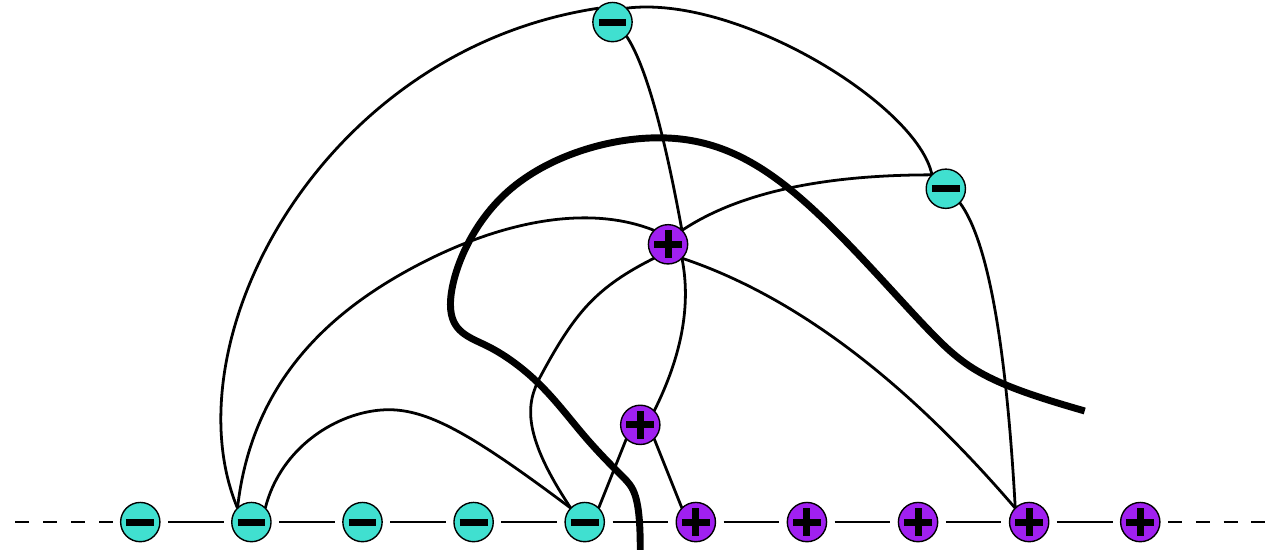}
\end{center}\vspace{-1em}
\caption{The first visit of the interface to the \+ boundary and the hull $\mathcal{H}$. In this example, $\partial\mathcal{H}=3$.}\label{fig:hull}
\end{figure}

For $\eta_{p,q}$, we have \begin{align}\label{eq:scalingintpq}
\Prob_{p,q}(\eta_{p,q}- T_m>\delta p)	\ &=	 \	\EE_{p,q}\m[{ \Prob_{P_{T_m},Q_{T_m}}(\eta_{P_{T_m},Q_{T_m}} >\delta p) }
				\notag	\\ &\le\	\EE_{p,q}\m[{\sum_{p'=0}^m\Prob_{p',Q_{T_m}}(\eta_{p',Q_{T_m}} >\delta p)+\sum_{q'=0}^m\Prob_{P_{T_m},q'}(\eta_{P_{T_m},q'}>\delta p)}.
\end{align}
Let $M>0$ be large, and fix $p'\leq m$. Then we may estimate \begin{align*}
&\Prob_{p',Q_{T_m}}(\eta_{p',Q_{T_m}} >\delta p)\\&=\Prob_{p',Q_{T_m}}(\eta_{p',Q_{T_m}} >\delta p, \ Q_{T_m}>M)+\Prob_{p',Q_{T_m}}(\eta_{p',Q_{T_m}} >\delta p| \ Q_{T_m}\leq M)\Prob_{p',Q_{T_m}}(Q_{T_m}\leq M)\\ &\leq\Prob_{p',Q_{T_m}}(\eta_{p',Q_{T_m}} >\delta p, \ Q_{T_m}>M)+\max_{q'\leq M}\Prob_{p',q'}(\eta_{p',q'}>\delta p).
\end{align*} 
We have $\Prob_{p',q}\cv[]q\Prob_{p'}$. Thus, the first term in the previous expression can be bounded from above by $\Prob_{p'}(\eta_{p'}>\delta p)+\epsilon'$ for any $\epsilon'>0$,  provided $M$ is large enough, and the first of the aforementioned terms is already shown to converge to zero. The last term converges to zero as $p\to\infty$, since $\eta_{p',q'}$ is a.s. finite under $\prob_{p,q}$. The second sum in \eqref{eq:scalingintpq} treated similarly by symmetry. It follows that $\Prob_{p,q}(\eta_{p,q}-T_m >\delta p)\cv[]{p,q} 0$, proving the claim. \qed

We finally state a result concerning the scaling of the perimeter of a hull containing the portion of the interface before its first boundary visit at $\nu\in(1,\nu_c)$, which is retrieved from \cite{ACpercopeel} from the case of the critical percolation ($\nu=1$) on the vertices of the type I random triangulation. We just note that the qualitative behavior of the peeling process is the same for all $\nu\in[1,\nu_c)$, which allows us to deduce the claim. More precisely, let $\mathcal{H}$ be the hull which is composed of the explored part $\emap_n$ under $\Prob_\infty$ at the time when the process hits the \+ boundary first time, and let $\partial\mathcal{H}$ be its outer boundary, consisting of the edges boundary edges which are not on the boundary of the half-plane. Denote by $|\partial\mathcal{H}|$ the number of vertices of this boundary, which is a priori simple (see \cite{ACpercopeel}, where $\mathcal{H}$ is called the \emph{extended hull}). Then, we use Lemma \ref{lem:orderparam} and Proposition \ref{prop:stable} together with \cite{ACpercopeel} to deduce the following:

\begin{prop}
Let $\nu\in (1,\nu_c)$. Then, \begin{displaymath}
\prob_\infty^\nu(|\partial\mathcal{H}|>n)=n^{-1/2+o(1)}.
\end{displaymath}
\end{prop}

We leave it for future work to study the scaling of the boundary of a finite cluster, for which it is instructive to begin the peeling exploration from an infinite monochromatic boundary. In that case, the perimeter process will no longer be a random walk, and its distribution will depend on the length of the active boundary. 


\bibliographystyle{abbrv}
\bibliography{database-Joonas}

\end{document}